\documentclass{LMCS}

\def\dOi{11(1:20)2015}
\lmcsheading%
{\dOi}
{1--33}
{}
{}
{Feb.~28, 2014}
{Mar.~31, 2015}
{}

\ACMCCS{[{\bf Theory of computation}]: Semantics and
  reasoning---Program reasoning---Program verification; [{\bf Software
      and its engineering}]: Software organization and
  properties---Software functional properties---Formal
  methods---Software verification}


\usepackage{amsmath,amssymb}
\usepackage{color}
\usepackage[noend]{algpseudocode}

\usepackage{tikz}
\usetikzlibrary{shapes.symbols}

\algrenewcommand\algorithmicindent{1.0em}

\newcommand{\COMMENT}[1]{}
\newcommand{\defeq}{\;\stackrel{\rm def}{=}\;}

\newcommand{\NULL}{\ensuremath{\mathtt{NULL}}}

\newcommand{\enq}{\ensuremath{\mathtt{enq}}}
\newcommand{\deq}{\ensuremath{\mathtt{deq}}}
\newcommand{\false}{\mathsf{false}}
\newcommand{\true}{\mathsf{true}}
\newcommand{\assume}{\mathsf{assume}}
\newcommand{\Prg}{\mathit{Prg}}
\newcommand{\ndchoice}{\sqcup}
\newcommand{\bigndchoice}{\bigsqcup}

\newcommand{\remPending}[1]{\ensuremath{\mathit{remPending}({#1})}}
\newcommand{\Match}{\ensuremath{\mathit{Match}}}
\newcommand{\Compl}[1]{\ensuremath{\mathit{Compl}({#1})}}
\newcommand{\Deq}[1]{\ensuremath{\mathit{Deq}({#1})}}
\newcommand{\Enq}[1]{\ensuremath{\mathit{Enq}({#1})}}
\newcommand{\Before}[2]{\ensuremath{\mathit{Before}({#1},{#2})}}
\newcommand{\After}[2]{\ensuremath{\mathit{After}({#1},{#2})}}
\newcommand{\Val}[2]{\ensuremath{\mathit{Val}_{#1}({#2})}}

\newcommand{\uid}{\ensuremath{\mathit{uid}}}
\newcommand{\ltsq}{\ensuremath{{\mathsf{LTS}}_Q}}
\newcommand{\qtrans}[1]{\ensuremath{\xrightarrow{#1}}}
\newcommand{\compl}{\ensuremath{\hat}}
\newcommand{\mus}{\mu_{\rm seq}}
\newcommand{\permeq}{\sim_{\rm perm}}
\newcommand{\seqx}[2]{\ensuremath{#1\langle #2\rangle}}
\newcommand{\enqset}{\ensuremath{\mathtt{Enq}}}
\newcommand{\deqset}{\ensuremath{\mathtt{Deq}}}
\newcommand{\qbehavset}{\ensuremath{\mathcal{Q}}}
\newcommand{\qbehavsetx}[1]{\ensuremath{\qbehavset_{#1}}}
\newcommand{\obsequiv}{\ensuremath{\equiv_{\mathit{obs}}}}

\newcommand{\Bad}[2]{\ensuremath{\mathit{Bad}(#1,#2)}}
\newcommand{\Badx}[3]{\ensuremath{\mathit{Bad}_{#3}(#1,#2)}}
\newcommand{\dhat}{\ensuremath{d_{\bot}}}
\newcommand{\Dhat}{\ensuremath{D_{\dhat}}}
\newcommand{\Ehat}{\ensuremath{E_{\dhat}}}

\newcommand\inpar[1]{\left(\begin{array}{@{}l@{}}#1\end{array}\right)}
\newcommand\inatom[1]{\left\langle\begin{array}{@{}l@{}}#1\end{array}\right\rangle}

\newcommand\mylabel[1]{\label{#1}}

\newcommand{\VFresh}{\ensuremath{\mathsf{VFresh}}}
\newcommand{\VRepet}{\ensuremath{\mathsf{VRepet}}}
\newcommand{\VOrd}{\ensuremath{\mathsf{VOrd}}}
\newcommand{\VWit}{\ensuremath{\mathsf{VWit}}}
\newcommand{\POrd}{\ensuremath{\mathsf{POrd}}}
\newcommand{\PWit}{\ensuremath{\mathsf{PWit}}}

\begin{document}

\title[Aspect-Oriented Linearizability Proofs]{Aspect-Oriented Linearizability Proofs\rsuper*}

\author[S.~Chakraborty]{Soham Chakraborty\rsuper a}
\address{{\lsuper{a,d}}MPI-SWS, Kaiserslautern and Saarbr\"ucken, Germany}
\email{\{sohachak,viktor\}@mpi-sws.org}

\author[T.~A.~Henzinger]{Thomas A. Henzinger\rsuper b}
\address{{\lsuper b}IST Austria, Klosterneuburg, Austria}
\email{tah@ist.ac.at}

\author[A.~Sezgin]{Ali Sezgin\rsuper c}
\address{{\lsuper c}University of Cambridge Computer Laboratory, Cambridge, U.K.}
\email{as2418@cam.ac.uk}

\author[V.~Vafeiadis]{Viktor Vafeiadis\rsuper d}
\address{\vspace{-18 pt}}
%\email{viktor@mpi-sws.org}

\titlecomment{{\lsuper*}A preliminary version of this article appeared at CONCUR'13.}

\begin{abstract}
Linearizability of concurrent data structures is usually proved by monolithic
simulation arguments relying on the identification of the so-called linearization points.
Regrettably, such proofs, whether manual or automatic, are often complicated
and scale poorly to advanced non-blocking concurrency patterns, such as helping
and optimistic updates.

In response, we propose a more modular way of checking linearizability of
concurrent queue algorithms that does not involve identifying linearization
points.  We reduce the task of proving linearizability
with respect to the queue specification to establishing four basic
properties, each of which can be proved independently by simpler
arguments.  As a demonstration of our approach, we verify the Herlihy and Wing
queue, an algorithm that is challenging to verify by a simulation
proof.  

\end{abstract}

\maketitle

% !TEX root = main.tex

\section{Introduction}
\label{sec:introduction}

Linearizability~\cite{HW1990} is widely accepted as the standard correctness
requirement for concurrent data structure implementations.
It amounts to showing that each method provides the illusion that it {\em executes} atomically at some point after its call and before its return.
Typically, what each method is expected to do (atomically) is given in terms of a sequential specification.
For instance, an unbounded queue must support the following two methods: 
\emph{enqueue}, which extends the queue by appending one element to its end, 
and \emph{dequeue}, which removes and returns the first element of the queue.

The standard way to prove that a concurrent queue implementation is linearizable is
to show that it is simulated by the idealised atomic queue implementation, 
which we take to be the specification of the queue.
For example, using forward simulation~\cite{LV1995}, we have to define 
a relation $S$ relating the state of the implementation to the state of the specification,
and to show that 
(1) the initial states of the implementation and the specification are related by $S$, and 
(2) starting from $S$-related implementation and specification states $(\sigma_{\sf impl}, \sigma_{\sf spec})$, 
if the implementation takes a step and goes to state $\sigma'_{\sf impl}$, 
the specification can also take a matching step (or stutter) and 
result in some state $\sigma'_{\sf spec}$ that is $S$-related to $\sigma'_{\sf impl}$.
The most important part of these proofs is to decide which of the implementation steps
are matched by actual steps of the specification code and which by stuttering moves.
For each method of the implementation, the step during its execution that in the 
simulation proof is matched by the atomic step of the corresponding method of the 
specification is known as the \emph{linearization point}.
A well-established approach (e.g.~\cite{AHH+2013,ARR+2007,CDG2005,DSW2011,DM2009,LCL+2009,SWD2012,Vaf2009,Vaf2010})
is therefore to identify these linearization points, 
which when performed by the implementation change the state of the specification, 
and to then construct a suitable forward or backward simulation.

While for a number of concurrent algorithms, spotting the
linearization points may be straightforward (and has even been automated to some extent~\cite{Vaf2010}), 
in general specifying the linearization points can be very difficult. 
For instance, in implementations using a helping mechanism, they can lie in code not syntactically belonging to the thread and operation in question, and can even depend on future behavior.
There are numerous examples in the literature, where this is the case;
to mention only a few concurrent queues: 
the Herlihy and Wing queue~\cite{HW1990},
the optimistic queue~\cite{LS2004}, 
the elimination queue~\cite{MNS+2005}, 
the baskets queue~\cite{HSS2007}, 
the flat-combining queue~\cite{HIS+2010}.

\subsection*{The Herlihy and Wing Queue}\hfill
\label{subsec:HWQ}
\begin{figure}[h!]
\centering\begin{tabular}{@{}l@{~~}|@{~~}l@{}}
\begin{minipage}[t]{.37\textwidth}
\begin{algorithmic}[1]
 \State $\textbf{var}\ q.back : int \gets 0$
 \State $\textbf{var}\ q.items : \textbf{array of } val $
 \Statex \qquad$\gets \{ \NULL, \NULL, \ldots \}$
 \Statex
 \Procedure{$\enq$}{$x : val$}
  \State $\langle{i \gets \mathtt{INC}(q.back)}\rangle$
  \hfill $\triangleright\enspace\! E_1$
   \State $\left\langle{q.items[i]\gets x}\right\rangle$
   \hfill $\triangleright\enspace\! E_2$
 \EndProcedure
\algstore{mybreak}
\end{algorithmic}
\end{minipage}
&
\begin{minipage}[t]{.55\textwidth}
\vskip 0em
\begin{algorithmic}[1]
\algrestore{mybreak}
 \Procedure{$\deq$}{{}} $: val$
  \While{true}
   \State $\langle{range\gets q.back-1}\rangle$ 
   \hfill $\triangleright\enspace\! D_1$
  \For{$i=0$ \textbf{to} $range$}
    \State $\left\langle{x\gets \mathtt{SWAP}(q.items[i],\NULL)}\right\rangle$
   \hfill $\triangleright\enspace\! D_2$
   \If{$x \neq\NULL$} \Return $x$
   \EndIf
  \EndFor
 \EndWhile
 \EndProcedure
\end{algorithmic}
\end{minipage}
\end{tabular}
\caption{Herlihy and Wing queue~\cite{HW1990}.}\label{fig:hw-queue}
\end{figure}

In this paper, we focus on the Herlihy and Wing queue~\cite{HW1990} (henceforth, HW queue for short)
that illustrates nicely the difficulties encountered when defining a simulation relation based on linearization points.
We recall the code of the queue as given in \cite{HW1990} in Figure~\ref{fig:hw-queue}.
The queue is represented as a pre-allocated unbounded array, $q.items$, initially filled with $\NULL$s, 
and a marker, $q.back$, pointing to the end of the used part of the array.
Enqueuing an element is done in two steps: the marker to the end of the array is incremented ($E_1$), thereby reserving a slot for storing the element, and then the element is stored at the reserved slot ($E_2$).
Dequeue is more complex: it reads the marker ($D_1$), and then searches from the beginning of the array up to the marker to see if it contains a non-$\NULL$ element.
It removes and returns the first such element it finds ($D_2$). 
If no element is found, dequeue starts again afresh.
Each of the four statements surrounded by $\langle\rangle$ brackets and annotated 
by $E_i$ or $D_i$ for $i=1,2$ is assumed to execute atomically.

We now show that verifying this algorithm by finding its linearization points is difficult.
Consider the following execution fragment, where $\cdot$ denotes context switches between concurrent threads,
\[
(t:E_1)\cdot (u:E_1) \cdot (v:D_1,D_2)\cdot (u:E_2) \cdot (t:E_2) \cdot (w:D_1)
\]
which have threads $t$ and $u$ executing enqueue instances, $v$ and $w$ executing dequeue instances. 
At the end of this fragment, $v$ is ready to dequeue the element enqueued by $u$, and $w$ is ready to dequeue the element enqueued by $t$.
In order to define a simulation relation from this interleaving sequence to a valid sequential queue behavior, where operations happen in isolation, we have to choose the linearization points for the two completed enqueue instances. 
The difficulty lies in the fact that no matter which statements are chosen as the linearization points for the two enqueue instances, there is always an extension to the fragment inconsistent with the particular choice of linearization points. 
For instance, if we choose $(t:E_1)$ as the linearization point for $t$, then the extension 
\[
(v:D_2,\textbf{return}) \cdot (z:D_1,D_2,\textbf{return})
\]
requiring $u$'s element be enqueued before that of $t$'s, will be inconsistent.
If, on the other hand, we choose any statement which makes $u$ linearize before $t$, then the extension
\[
(w:D_2,\textbf{return}) \cdot (z:D_1,D_2,D_2,\textbf{return})
\]
requiring the reverse order of enqueueing will be inconsistent. 
This shows not only that finding the correct linearization points can be challenging,
but also that the simulation proofs will require to reason about the entire state of the
system,
as the local state of one thread can affect the linearization of another.

\subsection*{Our Contribution}

In our experience, this and similar tricks for reducing synchronization among threads so as to achieve better performance, make concurrent algorithms extremely difficult to reason about when one is constrained to establishing a simulation relation.
However, if two methods overlap in time, then the only thing enforced by linearizability is that their effects are observed in {\em some} and same order by all threads.
For instance, in the example given above, the simple answer for the particular ordering between the linearization points of the enqueue instances of $t$ and $u$, is that it does not matter!
As long as enqueue instances overlap, their values can be dequeued in any order.

Building on this observation, our contribution is to simplify linearizability proofs by modularizing them.
We reduce the task of proving linearizability to establishing four relatively
simple properties, each of which may be reasoned about independently.
In (loose) analogy to aspect-oriented programming, 
we are proposing ``aspect-oriented'' linearizability proofs for concurrent queues,
where each of these four properties will be proved independently.

So what are these properties?  A correct (i.e., linearizable)
concurrent queue:
\begin{enumerate}
\item must not allow dequeuing an element that was never enqueued;
\item must not allow the same element to be dequeued twice;
\item must not allow elements to be dequeued out of order; and
\item must correctly report whether the queue is empty or not.%
\footnote{The HW queue trivially satisfies the fourth property as it never reports
that the queue is empty.}
\end{enumerate}

Although similar properties were already mentioned by Herlihy and Wing~\cite{HW1990}, 
we for the first time prove that suitably formalized versions of these four properties
are not only necessary, but also sufficient, conditions for linearizability 
with respect to the queue specification, at least for what we call 
\emph{purely-blocking} implementations.
This is a rather weak requirement satisfied by all non-blocking implementations,
as well as by possibly blocking implementations, such as HW $\deq()$ method, whose
blocking executions do not modify the global state.

\subsection*{Paper Outline}

The rest of the paper is structured as follows:
Section~\ref{sec:back} recalls the definition of linearizability in terms of execution histories.
Section~\ref{sec:seqwit} develops an alternative characterization of legal queue behaviors, which is useful for our proofs.
Section~\ref{sec:conditions} formalizes the aforementioned four properties, and proves that they are necessary and sufficient conditions for proving linearizability of queues.
Section~\ref{sec:herlihy-wing} returns to the HW queue example and presents a detailed manual proof of its correctness by checking each of the properties separately.
Section~\ref{sec:checking} shows how the checking of these four properties can be automated by reducing them to non-termination of certain parametric programs.
Section~\ref{sec:cave} explains how we adapted \textsc{Cave}~\cite{Vaf2010} to prove these parametric programs non-terminating for the case of the HW queue.
Finally, in Section~\ref{sec:related-work} we discuss related work, and in Section~\ref{sec:conclusion} we conclude.

\subsection*{Differences from the Conference Paper}

This article is an extended version of our CONCUR'13 conference paper~\cite{HSV2013},
containing all the proofs of the lemmas and theorems mentioned in the paper.
Since the conference, we have also implemented a checker for the \VRepet\ property,
and have expanded the discussion of automation in Sections~\ref{sec:checking} 
and~\ref{sec:cave} to cover the verification of \VRepet.

% !TEX root = main.tex

\section{Technical Background}
\label{sec:back}

In this section, we introduce common notations that will be used throughout the paper and recall the definition of linearizability.

For any function $f$ from $A$ to $B$ and $A'\subseteq A$, let $f(A') \defeq \{f(a) \mid a\in A'\}$.
Given two sequences $x$ and $y$, let $x\cdot y$ denote their concatenation, and let $x\permeq y$ hold if one is a permutation of the other.
We use $\seqx x i$ to refer to the $i^{th}$ element in sequence $x$,
and $\seqx x {i:j}$ to refer to the subsequence of $x$ containing all 
elements from position $i$ to $j$ inclusive.
We write $x|A$ for the subsequence of $x$ containing only
elements in the set $A$.

\subsection*{Behaviors}

A {\em data structure} $\mathcal{D}$ is a pair $(D,\Sigma_{\mathcal{D}})$, where $D$ is the {\em data domain} and $\Sigma_{\mathcal{D}}$ is the {\em method alphabet}.
An {\em event} of $\mathcal{D}$ is a quadruple $(\uid,m,d_i,d_o)$, for a unique event identifier, $\uid\in\mathbb{N}$, a method $m\in\Sigma_{\mathcal{D}}$, and data elements $d_i,d_o\in D$.
Intuitively, $(\uid,m,d_i,d_o)$ denotes the application of method $m$ with input argument $d_i$ returning the output value $d_o$.
Throughout the paper, we will assume that the $\uid$ components of events are globally unique.
A duplicate-free sequence over events of $\mathcal{D}$ is called a {\em behavior}.
The {\em semantics} of data structure $\mathcal{D}$ is a set of behaviors, called {\em legal} behaviors.

The method alphabet $\Sigma_Q$ of a queue is the set $\{\enq,\deq\}$.
We will take the data domain to be the set of natural numbers, $\mathbb{N}$, and a distinguished symbol {\NULL} not in $\mathbb{N}$.
Events are written as $\enq^\uid(x)$, short for $(\uid,\enq,x,\bot)$, and $\deq^\uid(x)$, short for $(\uid,\deq,\bot,x)$. For consiceness, we will often also omit the $\uid$ superscripts.
Events with {\enq} are called {\em enqueue} events, and those with {\deq} are called {\em dequeue} events.
We use $\enqset$ and $\deqset$ to denote all enqueue and dequeue events, respectively.

We will use a labelled transition system, $\ltsq$, to define the queue semantics.
The states of {\ltsq} are sequences over $\mathbb{N}$, the initial state is the empty sequence $\varepsilon$. 
There is a transition from $q$ to $q'$ with action $a$, written $q\qtrans{a}q'$, if
(\textit{i}) $a=\enq(x)$ and $q'=q\cdot x$, or 
(\textit{ii}) $a=\deq(x)$ and $q=x\cdot q'$, or 
(\textit{iii}) $a=\deq(\NULL)$ and $q=q'=\varepsilon$.

A {\em run} of {\ltsq} is an alternating sequence $q_0l_1q_1\ldots l_nq_n$ of states and queue events such that for all $1\leq i\leq n$, we have $q_{i-1}\qtrans{l_i}q_i$.
The trace of a run is the sequence $l_1\ldots l_n$ of the events occurring on the run.
A queue behavior $b$ is \emph{legal} iff there is a run of {\ltsq} with trace $b$.
In what follows, we will consider only legal queue behaviors, and hence usually omit {\em legal}, unless explicitly stated otherwise.
Let $\qbehavset$ denote the set of all (legal) queue behaviors.

\subsection*{Histories and Linearizability}

Each event $a=(\mathit{uid},m,d_i,d_o)$ generates two {\em actions}: the {\em invocation} of $a$, written as $inv(a)$, and the {\em response} of $a$, written as $res(a)$.
We will also use $m^\uid_i(d_i)$ and $m^\uid_r(d_o)$ to denote the invocation and the response actions, respectively. 
When a particular method $m$ does not have an input (resp., output) parameter, we will write
$m^\uid_i$ (resp., $m^\uid_r$) for the corresponding invocation (resp., response) action. 
We will also often omit the superscripts, when they are not important.

In this paper, a {\em history} of $\mathcal{D}$ is a sequence of invocation and response actions of $\mathcal{D}$. 
We will assume the existence of an implicit identifier in each history $c$ that uniquely pairs each invocation with its corresponding response action, if the latter also occurs in $c$.
A history $c$ is {\em well-formed} if every response action occurs after its associated invocation action in $c$.
We will consider only well-formed histories.
An event is {\em completed} in $c$, if both of its invocation and response actions occur in $c$.
An event is {\em pending} in $c$, if only its invocation occurs in $c$.
We define $\remPending c$ to be the sub-sequence of $c$ where all pending events have been removed.
An event $e$ precedes another event $e'$ in $c$, written $e\prec_ce'$, if the response of $e$ occurs before the invocation of $e'$ in $c$.
For event $e$, $\Before e c$ denotes the set of all events that precede $e$ in $c$. 
Similarly, $\After e c$ denotes the set of all events that are preceded by $e$ in $c$.
Formally, 
\[\Before e c \defeq \{e' \mid e'\prec_c e\} \qquad\text{and}\qquad
  \After e c  \defeq \{e' \mid e \prec_c e'\} \,. \]
A set of events $A$ is {\em closed under $\prec_c$} iff whenever $a\in A$ and $b\prec_c a$, then $b\in A$.

History $c$ is called {\em complete} if it does not have any pending events.
For a possibly incomplete history $c$, a {\em completion} of $c$, written $\compl{c}$, is a well-formed complete history such that $\compl{c}=\remPending{c\cdot c'}$ where $c'$ contains only response actions.
Let $\Compl c$ denote the set of all completions of $c$.

A history is called {\em sequential} if all invocations in it are immediately followed by their matching responses, with the possible exception of the very last action which can only be the invocation of a pending event.
We identify complete sequential histories with behaviors of $\mathcal{D}$ by mapping each consecutive pair of matching actions in the former to its event constructing the latter.
A sequential history $s$ is a {\em linearization} of a history $c$, if there exists $\compl{c}\in\Compl c$ such that $\compl{c}\permeq s$ and whenever $e\prec_{\compl{c}}e'$ we have $e\prec_se'$.

\begin{defi}[Linearizability~\cite{HW1990}]
A history $c$ is linearizable with respect to a data structure $\mathcal{D}$ if
there exists a linearization of $c$ that is a legal behavior of $\mathcal{D}$.
A set of histories $C$ is linearizable with respect to $\mathcal{D}$ if 
every $c\in C$ is linearizable with respect to $\mathcal{D}$.
\end{defi}
An {\em execution trace} is a sequence of instruction labels coupled with thread identifiers executing the instruction.
For instance, $(t:i)$ denotes the execution of instruction with the unique label $i$ by thread $t$.
An instruction label is the {\em entry} point of method $m$, written $enter(m)$, if it is the label of the first instruction of $m$.
Similarly, an instruction label is an {\em exit} point of $m$, written $exit(m)$, if it is the label of an instruction that completes the execution of $m$. 
Each execution trace $\tau$ induces a history $h(\tau)$ which is obtained by replacing each $(t:enter(m))$ with $m^{t_{uid}}_i(d_i)$, each $(t:exit(m))$ with $m^{t_{uid}}_r(d_o)$, and removing the remaining symbols.
We assume that states of an execution trace contain enough information to deduce the values of $d_i$ and $d_o$ associated with each entry and exit point. 
To illustrate this definition, consider the following execution trace from the introduction:
\[
(t:E_1)\cdot(u:E_1)\cdot(v:D_1,D_2)\cdot(u:E_2)\cdot(t:E_2)
       \cdot(w:D_1)\cdot(v:D_2,\textbf{return})\cdot(z:D_1,D_2,\textbf{return})
\,.
\]
The history corresponding to this trace is:
\[ 
\enq^t_i(x) \cdot \enq^u_i(y)\cdot \deq^v_i \cdot
\deq^w_i \cdot \deq^v_r(x) \cdot  \deq^z_i \cdot \deq^z_r(y)
\]
where we have used the thread identifiers without subscripts as unique event identifiers.
After completing the history with responses $\enq^t_r$ and $\enq^u_r$ of the pending enqueues, 
and removing the pending invocation $\deq^w_i$, the history may be linearized as follows:
\[
\enq^t_i(x) \cdot \enq^t_r \cdot \enq^u_i(y) \cdot \enq^u_r \cdot 
\deq^v_i\cdot \deq^v_r(x) \cdot \deq^z_i \cdot \deq^z_r(y)
\]
and corresponds to the (legal) behavior
$\enq^t(x) \cdot \enq^u(y) \cdot \deq^v(x) \cdot \deq^z(y)$.
An execution trace is {\em complete} if its induced history is complete.
An implementation is identified with the set of execution traces it generates. 
When clear from the context, we will refer to the induced history of an execution trace as a history of the implementation.

\section{Alternative Characterization of Legal Queue Behaviours}
\label{sec:seqwit}

We start with some terminology.
Let $c$ be a history. 
$\Enq c$ denotes the set of all enqueue events invoked (and not necessarily completed) in $c$.
Similarly, $\Deq c$ denotes the set of all dequeue events invoked in $c$.
When $c$ is a complete history, we define the value of an event $e$, written $\Val c e$, to be
the value enqueued or dequeued by that event.

We find it useful to express the semantics of queues in an alternative formulation.

\begin{defi}\mylabel{def:seqwit}
A queue behavior $b$ has a {\em sequential witness} if there is a total mapping $\mus$ from $\Deq b$ to $\Enq b\cup\{\bot\}$ such that
\begin{enumerate}[label=(\roman*)]
\item $\mus(d)=e$ implies $\Val b d = \Val b e$, 
\item $\mus(d)=\bot$ iff $\Val b d=\NULL$,
\item $\mus(d)=\mus(d')\neq\bot$ implies $d=d'$,
\item $\mus(d)=e$ implies $e\prec_b d$,
\item $e\prec_b \mus(d')$ implies $\mus^{-1}(e)\prec_b d'$,
\item $\mus(d)=\bot$ implies $|\{e\in\Enq b \mid e\prec_bd\}|=|\{d'\in\Deq b \mid d' \prec_b d\wedge\mus(d') \neq\bot\}|$.
\end{enumerate}
\end{defi}

\noindent To illustrate this definition, consider the following legal queue behavior:
\[
  b \defeq    \enq^t(x) \cdot \enq^u(y) \cdot \deq^v(x)  \cdot \deq^z(y) \cdot \deq^w(\NULL)
\]
We can pick $\mus$ such that $\mus(v)=t$, $\mus(z)=u$ and $\mus(w)=\bot$.
The constraints of Definition~\ref{def:seqwit} are satisfied because:
\begin{itemize}
\item $t\prec_b u$ implies $\mus^{-1}(t)=v \prec_b z=\mus^{-1}(u)$.
\item 
$|\{e\,{\in}\,\Enq b \mid e\prec_b w\}|=|\{t,u\}|=|\{v,z\}|=|\{d'\,{\in}\,\Deq b \mid d'\prec_b w\land \mus(d')\,{\neq}\,\bot\}|$.
\end{itemize}

To show that a behavior is legal iff it has a sequential witness, we need a number of 
auxiliary definitions and lemmas.

We say that two queue behaviors $b_1$ and $b_2$ are {\em observationally equivalent}, written $b_1 \obsequiv b_2$, 
if the sequences of enqueue events and those of dequeue events agree in both behaviors.
Formally, $b_1 \obsequiv b_2$ iff $b_1 | \enqset = b_2 | \enqset$ and $b_1 | \deqset = b_2 | \deqset$.

We define a special subset of queue behaviors, the {\em canonical} subset $\qbehavsetx c$, in which each enqueue event is immediately followed by its matching dequeue event, in case it exists.
Formally, the canonical queue behaviors are given by the following regular expression:
\[
\qbehavsetx c \stackrel{\rm def}{=} 
\left( (\deq(\NULL))^* \cdot \textstyle\sum_{x\in\mathbb{N}}\enq(x) \cdot \deq(x) \right)^* 
  \cdot (\deq(\NULL))^* \cdot \left(\textstyle\sum_{x\in\mathbb{N}}\enq(x)\right)^*
\]
A run $r$ of {\ltsq} is called {\em canonical} if the trace of $r$ is canonical.

Note that every canonical behavior is legal. Consider a canonical behavior $b \in \qbehavsetx c$ 
and split it into $b = b' \cdot \enq(v_1) \enq(v_2) \ldots \enq(v_n)$ 
such that $b'$ does not end with an $\enq$.
This behavior can be generated by the following run of the $\ltsq$.
\[
\begin{array}{l}
\left( (\epsilon\ \deq(\NULL))^*\ \epsilon\ \textstyle\sum_{x\in\mathbb{N}}\enq(x)\ x\ \deq(x) \right)^*\ 
  (\epsilon\ \deq(\NULL))^*\ \epsilon \\
  \enq(v_1) v_1 \enq(v_2) (v_1v_2) \cdots (v_1\ldots v_{n-1}) \enq(v_n) (v_1\ldots v_{n-1}v_n))
\end{array}
\]

As the following result shows, canonical queue behaviors represent all legal behaviors, up to observational equivalence.

\begin{lem}\mylabel{lem:canonical-representative}
Let $b\in\qbehavset$ be a legal queue behavior.
Then there exists a canonical behavior $b_c\in \qbehavsetx c$ such that $b \obsequiv b_c$.
\end{lem}

\proof
By induction on the length of $b$.
The base case, where $b=\varepsilon$ trivially satisfies the condition since $\varepsilon\in\qbehavsetx c$.
Now assume the claim holds for $b$, and we have to prove it for $b\cdot a$.
By the induction hypothesis, there is a canonical behavior $b_c \obsequiv b$.
We observe that since $b_c$ and $b$ are legal and $b_c \obsequiv b$, their runs 
end in the same final state. Therefore, the fact that $b\cdot a$ is legal implies
that $b_c\cdot a$ is also legal.
We proceed by a case analysis of $a$.
\begin{itemize}
\item $a=\enq(x)$. Then $b_c \cdot \enq(x)$ is trivially also canonical.
\item $a=\deq(\NULL)$.
We know that $b_c$ cannot end in an enqueue event, or else the queue would not
be empty. Therefore $b_c \cdot \deq(\NULL)$ is canonical.
\item $a=\deq(y)$, with $y\neq\NULL$.
We know that $b_c$ must be of the form $b^1_c \cdot \enq(y) \cdot b^2_c$ 
where $b^1_c$ does not end in a enqueue event and $b^2_c$ contains only 
enqueue events.
Then, we take the behavior $b^1_c \cdot \enq(y) \cdot \deq(y) \cdot b^2_c$,
which is both canonical and observationally equivalent to $b \cdot \deq(y)$.
\qed
\end{itemize}

\noindent Moreover, canonical behaviors have a straightforward sequential witness.

\begin{lem}\mylabel{lem:canonical-seqwit}
Every canonical behavior $b\in\qbehavsetx c$ has a sequential witness $\mu_{b}$.
\end{lem}
\proof
Let $b$ be a canonical behavior.
We construct $\mu_b$ by mapping all $\deq(\NULL)$ to $\bot$, and each $\deq(x)$ with $x\in\mathbb{N}$ to its immediate predecessor.
By the definition of canonical behavior, each $\deq(x)$ in $b$ is immediately preceded by $\enq(x)$.
Thus, the first four conditions are trivially satisfied.
If $\enq(y)\prec_b \enq(x)$ and $\deq(x)$ is in $b$, then by the definition of canonical behavior, we must have
\[
b=b_1\cdot \enq(y)\cdot \deq(y)\cdot b_2\cdot \enq(x)\cdot \deq(x)\cdot b_3
\]
for some sequences $b_1,b_2,b_3$.
This implies that condition (v) is satisfied.
Finally, if $b=b_1\cdot \deq(\NULL)\cdot b_2$, consider the sequence $b_1'$ obtained by projecting out all $\deq(\NULL)$ events from $b_1$.
That is, 
\[
b_1' \defeq b_1|(\enqset \cup \deqset\setminus\{\deq(\NULL)\})
\]
Then, by the definition of canonical behavior, we have
$
b_1'\in \left(\textstyle\sum_{x\in\mathbb{N}}\enq(x)\cdot\deq(x)\right)^*
$.
In other words, $b_1'$ has an equal number of $\enq$ and $\deq$ symbols, such that by construction $\mu_b(\deq(x))=\enq(x)\neq\bot$.
This implies that condition (vi) is satisfied.
\qed

\begin{lem}\mylabel{lem:deq-after-enq}
If $b=b_1\cdot \deq(x)\cdot b_2$  is a legal queue behavior, then $\enq(x)\prec_b\deq(x)$.
\end{lem}
\proof
By the definition of {\ltsq}, $\deq(x)$ can happen at a state $q$ if $q=x\cdot q'$ for some sequence $q'$.
Again by definition, all runs of {\ltsq} reaching $q$ must have a transition with label $\enq(x)$; otherwise, $x$ cannot occur in $q$.
Since all legal behaviors have a corresponding run, $\enq(x)\prec_b\deq(x)$ must hold.
\qed

Next, we show that observationally equivalent legal queue behaviors cannot reorder their $\deq(\NULL)$ events.

\begin{lem}\mylabel{lem:obsequiv-deqnull}
If $b$ and $b'$ are observationally equivalent legal queue behaviors, then $\seqx b i=\deq(\NULL)$ iff $\seqx {b'} i=\deq(\NULL)$.
\end{lem}
\proof
Since $\obsequiv$ is a symmetric relation, we prove only one direction.
($\Rightarrow$)
Consider the subsequences $b_e=\seqx b {1:i-1} | \enqset$, $b_d=\seqx b {1:i-1} | \deqset$, and their duals $b_e'$ and $b_d'$ for $b'$.
Note that each enqueue event increases by one the length of the sequence representing the state, each dequeue event decreases by one the length of the sequence representing the state, and $\deq(\NULL)$ can only happen when the length of the sequence is zero ($q=\varepsilon$).
Then, the number of enqueue events in $b_e$ and the number of non-{\NULL} dequeue events in $b_d$ must be equal; let us call it $k$.

Assume first that $b_d$ is a proper prefix of $b_d'$.
This implies that $b_e'$ is a proper prefix of $b_e$.
The $(|b_d|+1)^{th}$ symbol in $b_d'$ is $\deq(\NULL)$ because $b\obsequiv b'$.
Then, the number of non-{\NULL} dequeue events preceding this $\deq(\NULL)$ is $k$, but the number of enqueue events preceding it, contained in $b_e'$, which is a proper prefix of $b_e$, is strictly less than $k$.
This contradicts the assumption that $b'$ is a legal queue behavior.
The case where $b_e$ is a proper prefix of $b_e'$ follows a similar argument.

Now assume that $b_d'=b_d$ and $b'_e=b_e$.
Further assume $\seqx {b'} i=\enq(x)$ for some $x\in\mathbb{N}$.
Because $b\obsequiv b'$, the next dequeue event in $b'$ is necessarily $\deq(\NULL)$.
This, however, contradicts that the fact that $b'$ is legal, 
because in a legal behavior $\deq(\NULL)$ cannot immediately follow an enqueue event.
Therefore $\seqx {b'} i=\deq(y)$ for some $y$ and because $b\obsequiv b'$, we have that $y=\NULL$, as required.
\qed

Next, we show that given two observationally equivalent behaviors 
and a sequential witness for the first behavior, 
we can build a sequential witness for the other.

\begin{lem}\mylabel{lem:obsequiv-witness}
Let $b\obsequiv b'$ and $\mu_b$ be a sequential witness for $b$.
Then, there exists a sequential witness for $b'$.
\end{lem}
\proof
Let $\pi$ denote a permutation from $b$ to $b'$ such that $\deq(\NULL)$ events are not shuffled.
That is, $\pi(i)=j$ means that $\seqx {b} i=\seqx {b'} j$ and whenever $\seqx b i=\deq(\NULL)$, we set $\pi(i)=i$ By the definition of obs-equivalence and Lemma~\ref{lem:obsequiv-deqnull}, this permutation is well-defined.
Note that because $b \obsequiv b'$, the ordering among dequeue events is preserved by $\pi$.
That is, if $\seqx b i,\seqx b j\in\deqset$ and $i<j$, then $\pi(i)<\pi(j)$.
The same holds for the ordering among enqueue events.
We now pick the mapping $\mu(i) \defeq \pi(\mu_b(\pi^{-1}(i)))$ and show that it is a sequential witness for $b'$.
Conditions (i) and (ii) are satisfied by construction.
Condition (iii) is satisfied because $\pi$ is a bijection.
Condition (iv) is satisfied by Lemma~\ref{lem:deq-after-enq} and by the construction of $\mu$.
Condition (v) is satisfied by because $\pi$ is a bijection and preserves the ordering between dequeues and enqueues.
Condition (vi) is satisfied by Lemma~\ref{lem:obsequiv-deqnull}.
\qed

\begin{lem}\mylabel{lem:seqwit-removals}
Let $b$ be a queue behavior with a sequential witness $\mu$.
\begin{enumerate}
\item Let $d$ and $e$ be dequeue and enqueue events such that $\mu(d)=e$, and 
let $b'$ be the behavior obtained after removing both $d$ and $e$ from $b$. 
Then, the restriction of $\mu$ to $b'$ is a sequential witness for $b'$.
\item Let $d=\deq(\NULL)$, and $b'$ by obtained by removing $d$ from $b$.
Then, the restriction of $\mu$ to $b'$ is a sequential witness for $b'$.
\end{enumerate}
\end{lem}
\proof
In both cases, let us denote the restriction of $\mu$ on $b'$ with $\mu'$; 
we have to show that $\mu'$ satisfies the six conditions of a sequential witness.

(1)
Since $\mu'$ is a restriction, it satisfies conditions (i), (ii) and (iii).
Observe that for any two events $e_1$ and $e_2$ in $b'$, we have $e_1\prec_b e_2$ iff $e_1\prec_{b'} e_2$.
This implies that $\mu'$ satisfies conditions (iv) and (v).
Finally, we have to show that there cannot be a dequeue event $d'=\deq(\NULL)$ such that $e\prec_b d'\prec_b d$.
Assume the contrary, then since the number of non-{\NULL} dequeue events and the number of enqueue events preceding $d'$ must be equal, there must be a dequeue event $d_y=\deq(y)$ whose matching $e_y=\enq(y)$ comes after $d'$.
This implies that $d_y\prec_b d' \prec_b e_y$ and $\mu(d_y)=e_y$, contradicting condition (iii).
Thus, condition (vi) is also satisfied by $\mu'$.

(2)
As in the previous case, conditions (i) to (v) are satisfied by $\mu'$.
The condition (vi) is satisfied because removing $d'$ does not affect the cardinality of either set; thus, if $d'=\deq(\NULL)$ is in $b'$, then the number of enqueue events and non-{\NULL} dequeue events that precede $d'$ in $b'$ is the same as those that precede $d'$ in $b$.
\qed

\begin{lem}\mylabel{lem:seqwit-canonical}
Let $b$ be a queue behavior and let $\mu$ be a sequential witness for $b$.
Then, there exists a canonical behavior $b_c$ such that $b\obsequiv b_c$ and for all $i$, $\seqx b i=\deq(\NULL)$ iff $\seqx {b_c} i=\deq(\NULL)$.
\end{lem}
\proof
Let $\mathsf{can}(b,\mu)$ denote the canonical behavior of $b$ whose sequential witness is $\mu$.
We will prove, by induction on the length of $b$, that $\mathsf{can}$ is a well-defined total function.

For the base, consider all sequences $b$ of length 1 or less which have a sequential witness.
\begin{itemize}
\item If $b=\varepsilon$, then the empty mapping is the only sequential witness for $b$; by definition, $b$ is a canonical behavior.
The second condition is vacuously satisfied.
\item If $b=\deq(\NULL)$, then $\mu$ which maps $\deq(\NULL)$ to $\bot$ is the only sequential witness for $b$; by definition, $b$ is a canonical behavior.
Since $b_c=b$, the second condition is satisfied.
\item If $b=\enq(x)$ for some $x\in\mathbb{N}$, then the empty mapping is the only sequential witness for $b$; by definition, $b$ is a canonical behavior.
Since there is no $\deq(\NULL)$ event, the second condition is vacuously satisfied.
\end{itemize}
Observe that the sequence $b=\deq(x)$ of length 1 cannot have a sequential witness, because any sequential witness has to map $\deq(x)$ to a matching enqueue which does not exist.

Assume that the claim holds for all sequences of length $k$ or less.
Let $b$ be a sequence of length $k+1$ and $\mu$ be a sequential witness for $b$.
Consider the two sub-sequences of $b$, $b_d=b|\deqset$ and $b_e=b|\enqset$, with lengths $n_d$ and $n_e$, respectively.
Observe that $b$ is an interleaving of $b_d$ and $b_e$.
In particular, $\seqx b 1$ is either $\seqx {b_d} 1$ or $\seqx {b_e} 1$.
We will do a case analysis on the possible values for $d=\seqx {b_d} 1$.
\begin{itemize}
\item $d=\deq(\NULL)$.
Then, we set $b_c=\mathsf{can}(b,\mu)=d\cdot \mathsf{can}(b',\mu')$, with $b'$ obtained by removing the first $\deq(\NULL)$ from $b$ (note that this is $d$) and $\mu'$ obtained by restricting $\mu$ to $b'$.
By Lemma~\ref{lem:seqwit-removals}, $\mu'$ is a sequential witness for $b'$.
By inductive hypothesis, $b_c'=\mathsf{can}(b',\mu')$ is a canonical behavior observationally equivalent to $b'$.
Since $b_c'$ is a canonical behavior, so is $b_c=d\cdot b_c'=\deq(\NULL)\cdot b_c'$.
Since $b_c'\obsequiv b'$, we have $b_c|\deqset=d\cdot b_c'|\deqset=b_d=b|\deqset$, $b_c|\enqset=b_e=b|\enqset$.
Thus, $b_c\obsequiv b$.
The second condition is satisfied, because both $b$ and $b_c$ have $\deq(\NULL)$ in their first position and $b_c'$ preserves the positions of {\NULL}-dequeue events by inductive hypothesis.
\item $d=\deq(x)$ for some $x\in\mathbb{N}$.
By the assumption that $\mu$ is a sequential witness for $b$ implies that there exists $e=\enq(x)$ such that $\mu(d)=e$ (conditions (i) and (ii)) and $d\prec_b e$ (condition (iv)).
Then, $e=\seqx {b_e} 1=\enq(x)$ must hold.
Assume contrary, that is $\seqx {b_e} 1=\enq(y)$ for some $y\neq x$.
As noted above, $\seqx b 1$ is either $\seqx {b_d} 1$ or $\seqx {b_e} 1$.
If the former, then $d\prec_b e$ cannot hold since $d$ is minimal with respect to $\prec_b$, violating condition (iv) which contradicts the assumption that $\mu$ is a sequential witness for $b$.
If the latter, that is $e'=\seqx b 1=\seqx {b_e} 1=\enq(y)$, then $e'\prec_b e$, and either there is no $d'=\deq(y)$ or if it exists, $d\prec_b d'$, violating condition (v) which contradicts the assumption that $\mu$ is a sequential witness for $b$.
Thus, $e=\seqx {b_e} 1$. 
We set $b_c=\mathsf{can}(b,\mu)=e\cdot d\cdot \mathsf{can}(b',\mu')$, with $b'$ obtained by removing $d$ and $e$ from $b$ and $\mu'$ to be the restriction of $\mu$ on $b'$.
By Lemma~\ref{lem:seqwit-removals}, $\mu'$ is a sequential witness for $b'$.
By inductive hypothesis, $b_c'=\mathsf{can}(b',\mu')$ is a canonical behavior obs-equivalent to $b'$.
Since $b_c'$ is a canonical behavior, so is $b_c=e\cdot d\cdot b_c'=\enq(x)\cdot \deq(x)\cdot b_c'$.
Finally, since $b_c'\obsequiv b'$, we have $b_c|\deqset=d\cdot b_c'|\deqset=b_d=b|\deqset$ and $b_c|\enqset=e\cdot b_c'|\enqset=b_e=b|\enqset$. 
Thus, $b_c\obsequiv b$.
By the proof of Lemma~\ref{lem:seqwit-removals}, we know that for any $d'=\deq(\NULL)$ either both $d$ and $e$ precede it in $b$ or neither does.
Since $d$ is the first event in $b_d$, the latter cannot happen; i.e. $d\prec_b d'$ and $e\prec_b d'$.
This implies that the position of $d'$ is the same in $b_c$ and $e\cdot d\cdot b_c'$ by the inductive hypothesis.
Thus the second condition is satisfied. 
\qed
\end{itemize}

\begin{lem}\mylabel{lem:canonical-order-legal}
Let $b_c$ be a canonical queue behavior.
Let $b$ be a queue behavior such that $b\obsequiv b_c$, 
for every $\deq(x)$ in $b$ there is $\enq(x)\prec_b \deq(x)$, and 
for every $i$, $\seqx b {i}=\deq(\NULL)$ iff $\seqx {b_c} {i}=\deq(\NULL)$.
Then, $b$ is legal.
\end{lem}
\proof
We prove by induction on the length of $b$ that $b$ has a run in {\ltsq}.
The base case where $b=\varepsilon$ is trivial.
Assume that the claim holds for all sequences of length $k$ or less.
Let $b$ be a sequence of length $k+1$.
By the inductive hypothesis, there is a run $r$ in {\ltsq} with trace $\seqx b {1:k}$.
Let $q$ denote the state reached after this run.
It is enough to show that there is a transition in {\ltsq} of the form $q\qtrans{\seqx b {k+1}}q'$, for some $q'$.
We do a case analysis on $\seqx b {k+1}$.
\begin{itemize}
\item $\seqx b {k+1}=\enq(x)$. 
Then the desired transition is $q\qtrans{\enq(x)}q\cdot x=q'$.
\item $\seqx b {k+1}=\deq(x)=d$.
By the assumption on $b$, $e=\enq(x)\prec_b d$.
By observational equivalence to $b_c$, if $d=\seqx {(b|\deqset\setminus\{\deq(\NULL)\})} {i}$ for some $i$, then $e=\seqx {b|\enqset} {i}$.
Together they imply that there are exactly $i-1$ many non-{\NULL} dequeue events and at least $i$ many enqueue events that precede $d$ in $b$.
This in turn implies that $q$ must be of the form $x\cdot q'$.
Then the desired transition is $x\cdot q'\qtrans{\deq(x)}q'$.
\item $\seqx b {k+1}=\deq(\NULL)=d$.
By the assumption on $b$ and $b_c$, we have $\seqx {b_c} {k+1}=\deq(\NULL)$.
This implies that the number of enqueue events that occur in $\seqx {b_c} {1:k}$ is equal to the number of non-{\NULL} dequeue events in $\seqx {b_c} {1:k}$.
Since $b\obsequiv b_c$, for any dequeue event $d'$ we have $d'\prec_{b_c} d$ iff $d'\prec_b d$.
These in turn imply that for any enqueue event $e$ we have $e\prec_{b_c} d$ iff $e\prec_b d$.
Overall, we then have $q=\varepsilon$ and $\varepsilon\qtrans{\deq(\NULL)}\varepsilon=q'$ is the desired transition.
\qed
\end{itemize}

\begin{thm}\mylabel{thm:equiv-legal-seqwitness}
A queue behavior $b$ is legal iff $b$ has a sequential witness.
\end{thm}

\proof
($\Rightarrow$)
Let $b$ be a legal queue behavior.
By Lemma~\ref{lem:canonical-representative}, there is a canonical behavior $b_c$ such that $b_c\obsequiv b$.
By Lemma~\ref{lem:canonical-seqwit}, $b_c$ has a sequential witness.
By Lemma~\ref{lem:obsequiv-witness}, $b$ has a sequential witness.

($\Leftarrow$)
Let $b$ be a queue behavior and $\mu$ be a sequential witness for $b$.
By Lemma~\ref{lem:seqwit-canonical} and Lemma~\ref{lem:canonical-order-legal}, $b$ is legal.
\qed

% !TEX root = main.tex
\section{Conditions for Queue Linearizability}
\label{sec:conditions}

\subsection*{Generic Necessary and Sufficient Conditions}

We start by reducing the problem of checking linearizability of a given history, $c$, with respect to the queue specification to finding a mapping from its dequeue events to its enqueue events satisfying certain conditions.
Intuitively, we map each dequeue event to the enqueue event whose value the dequeue removed, or to nothing if the dequeue event returns \NULL.
We say that the mapping is {\em safe} if it pairs each {\deq} event with an {\enq} event such that the value removed by the former is inserted by the latter, 
implying that elements are inserted exactly once and removed at most once.
A safe mapping is {\em ordered} if it additionally respects the ordering of events in $c$.
Finally, an ordered mapping is a {\em linearization witness} if all {\NULL} returning {\deq} events see at least one state where the queue is logically empty.
Below, we formalize these notions.

\begin{defi}[Safe Mapping]\mylabel{def:safe}
A total mapping $\Match$ from $\Deq c$ to $\Enq c\cup\{\bot\}$ is {\em safe} for complete history $c$ if \\
(1) for all $d\in \Deq c$, if $\Match(d)\neq\bot$, then $\Val c d =\Val c {\Match(d)}$;\\
(2) for all $d\in \Deq c$, $\Match(d)=\bot$ iff $\Val c d=\NULL$; and\\
(3) for all $d, d'\in \Deq c$, if $\Match(d)=\Match(d')\neq\bot$, then $d=d'$. 
\end{defi}

\begin{defi}[Ordered Mapping]\mylabel{def:ordered}
A safe mapping $\Match$ for $c$ is {\em ordered} if \\
(1) for all $d\in \Deq c$, we have $d\not\prec_c \Match(d)$; and\\
(2) for all $e\in \Enq c$ and $d'\in \Deq c$, if $e\prec_c \Match(d')$, 
then there exists \mbox{$d\in \Deq c$} such that $e=\Match(d)$ and $d' \not\prec_c d$.
\end{defi}
Intuitively, the first condition states that an enqueue event cannot start after the completion of the dequeue event that removed the value inserted by the former.
The second condition states that if two enqueue events $e$ and $e'$ are ordered such that $e\prec_c e'$ and the value inserted by $e'$ is removed by some $d'$, then there must exist a dequeue event $d$ removing what $e$ has inserted and $d'$ cannot complete before $d$ starts.

Let $c$ be a complete history and $\Match$ be ordered for $c$.
Let $d_{\bot}\in \Deq c$ be a dequeue event returning {\NULL}; that is, $\Val c {d_{\bot}}=\NULL$. 
Define $\Bad c {d_{\bot}}\subseteq \Enq c$ as the smallest set consisting of all enqueue events $e$ in $c$ such that either if the matching dequeue $d$ for $e$ exists (i.e. $\Match(d)=e$), then $d$ is after $d_{\bot}$, or there is another $e'$ in $\Bad c {d_{\bot}}$ which precedes either $e$ or the matching dequeue event $d$ of $e$.
Formally, the definition is given inductively as follows:
\begin{align*}
\Badx c {\dhat} 0  = \{ e \in \Enq c \mid{} &
	\dhat \prec_c e ~\lor~ \forall d\in \Deq c\,.\, \Match(d)=e \Rightarrow \dhat \prec_c d\}\\
\Badx c {\dhat} {i+1} = \{ e \in \Enq c \mid{} &
	\exists e_i\in \Badx c {\dhat} i\,.\, e_i\prec_c e\\
& \qquad\qquad \lor \exists d\in \Deq c\,.\, \Match(d)=e \wedge e_i\prec_c d\}
\end{align*}
with
$
\Bad c {\dhat} = \cup_{i\in\mathbb{N}} \Badx c {\dhat} i 
$.

Intuitively, the set $\Bad c {d_{\bot}}$ contains all enqueue events after the completion of which $d_{\bot}$ cannot observe an empty queue.
In other words, if $e\in \Bad c {d_{\bot}}$ and if $e$ completes before $d_{\bot}$ does, then the state of the queue is guaranteed to be non-empty after $e$ completes until the completion of $d_{\bot}$. 

\COMMENT{
$e'\in \Enq c \cap \overline{\After c d}$ such that either there does not exist a $d'\in \Deq c$ such that $\Match(d')=e'$ or $d'$ exists and one of the following holds:
\begin{enumerate}
\item $d\prec_c d'$.
\item there exists $e''\in \Bad c d$ such that either $e''\prec_c e'$ or $e''\prec_c d'$.
\end{enumerate}
An equivalent definition given inductively is as follows:
}

\begin{defi}[Linearization Witness]\mylabel{def:lin-witness}
An ordered mapping $\Match$ for $c$ is a {\em linearization witness} if for any $d\in \Deq c$ with $\Val c d=\NULL$, we have $\Bad c d\cap \Before c d=\emptyset$.
\end{defi}

In the proofs that follow, we sometimes use the following result to prove that a given ordered mapping is a linearization witness.
\begin{lem}\mylabel{lem:alt-lin-witness}
Let $c$ be a complete history, $\Match$ be an ordered mapping for $c$ and $\dhat\in \Deq c$ be such that $\Match(\dhat)=\bot$.
Then, $\Bad c \dhat\cap \Before c \dhat=\emptyset$ iff there exist subsets $\Dhat\subseteq \Deq c$ and $\Ehat\subseteq \Enq c$ such that $(\Dhat\cup \Ehat)\cap \After c \dhat= \emptyset$, $\Dhat\cup \Ehat$ is closed under $\prec_c$, and $\Before c \dhat\cap \Enq c\subseteq \Ehat\subseteq \Match(\Dhat)$.
\end{lem}

\proof\hfill

\noindent {\bf ($\Rightarrow$)}
Assume that $\Bad c \dhat\cap \Before c \dhat= \emptyset$. Set
\[
\begin{array}{c@{}l}
\Ehat & \defeq \{ e\in \Enq c \mid e\notin(\After c \dhat\cup \Bad c \dhat)\}  \\
D'    & \defeq \{ d'\in \Deq c \mid \exists e\in \Ehat\,.\,\Match(d')=e\}  \\
\Dhat & \defeq D' \cup \{ d' \in \Deq c \mid \Match(d') = \bot \land \exists a\in \Ehat\cup D'.\  d'\prec_c a \} \\
\end{array}
\]
We have to show that $\Ehat$ and $\Dhat$ satisfy the three constraints.
\begin{itemize}
\item If $e\in \Ehat$, then it cannot be in $\After c \dhat$ by construction.
If $d\in \Dhat$, then either $d$ belongs to $D'$ or it is an event that precedes another event in $\Ehat\cup D'$.
If $d\in D'$, then by construction its matching $e=\Match(d)$ cannot be in $\Bad c \dhat$.
This implies that $\dhat\not\prec_c d$, hence $d\notin \After c \dhat$.
If $d\notin D'$, then it is in $\Dhat$ and there is some $d'$ such that $d\prec_c d'$ and $\dhat\not\prec_c d'$ which imply that $\dhat\not\prec_c d$, hence $d\notin \After c \dhat$.
\item Let $a'\in \Ehat\cup \Dhat$ and $a\prec_c a'$.
We do case analysis on $a$.
\begin{itemize}
\item If $a=d\in \Deq c$ with $\Match(d)=\bot$, then by the construction of $\Dhat$, $d\in \Dhat$.
\item If $a=d\in \Deq c$ with $\Match(d)\neq\bot$, then if there is $e\in \Ehat$ such that $\Match(d)=e$, then $d\in \Dhat$.
Assume that $\Match(d)=e\notin \Ehat$.
This can happen when either $e\in \After c \dhat$ or $e\in \Bad c \dhat$. 
If $e$ is in $\After c \dhat$, which by the assumption that $\Match$ is ordered implies that $d$ must complete after $e$ starts ($e\not\prec_c d$ must hold).
This in turn implies that $a'$, beginning after $d$ completes must be in $\After c {\dhat}$, which contradicts the assumption that $a'\in \Ehat$.
If $e$ is in $\Bad c \dhat$, then there must exist $e'\in\Bad c \dhat$ such that either $e'\prec_c e$ or $e'\prec_c d$.
Because $\Match$ is ordered, we have $d\not\prec_c e$.
Together with the assumption that $d\prec_c a'$, these imply $e'\prec_c a'$. 
Now, if $a'\in \Enq c$, then $a'\in \Bad c \dhat$ which contradicts the assumption that $a'\in \Ehat$.
If $a'\in \Deq c$ with $\Match(a')\neq\bot$, that $e'\in \Bad c \dhat$ and $e'\prec_c a'$ hold means that $\Match(a')\in \Bad c \dhat$ which in turn contradicts the assumption that $a'\in \Dhat$.
Finally, if $a'\in \Deq c$ with $\Match(a')=\bot$, then because $a'\in \Dhat$ there is some $d''\in D'$ such that $a'\prec_c d''$ which leads to the same contradiction as the previous case.
\item If $a=e\in \Enq c$, $a'$ is either an enqueue event $e'$ or there is a dequeue event $d'$ such that $e\prec_c d'$ and $\Match(d')\neq \bot$.
For the latter claim, observe that either $\Match(a')\neq\bot$ and we take $d'=a'$ or $\Match(a')=\bot$ and by definition of $\Dhat$ there exists $d'$ such that $a'\prec_c d'$ which by transitivity of $\prec_c$ implies $e\prec_c d'$.
If $e\notin \Ehat$, then either $e\in \Bad c \dhat$ or $e\in \After c \dhat$.
If $e\in \Bad c \dhat$ and $e\prec_c e'$ hold, then $e'$ must also be in $\Bad c \dhat$ contradicting the assumption that $e'\in \Ehat$.
If $e\in \Bad c \dhat$ and $e\prec_c d'$ hold, then $\Match(e')$, which exists because $\Match$ is safe, must be in $\Bad c \dhat$ which contradicts the assumption that $d'\in \Dhat$.
If $e\in \After c \dhat$, then $e\prec_c a'$ implies that $a'\in \After c \dhat$ contradicting the assumption that $a'\in \Ehat\cup \Dhat$.
\end{itemize}
Thus, we conclude that $a\in \Ehat\cup \Dhat$ whenever $a\prec_c a'$ for some $a'\in \Ehat\cup \Dhat$.
\item Let $e\in \Before c \dhat \cap \Enq c$.
By the assumption that $\Bad c \dhat \cap \Before c \dhat = \emptyset$, $e\notin \Bad c \dhat$.
Thus, by construction  $e\in \Ehat$, establishing $\Before c \dhat \cap \Enq c\subseteq \Ehat$.
Since $e\notin \Bad c \dhat$, there exists $d$ such that $\Match(d)=e$ and $d\in D'$, establishing $\Ehat \subseteq \Match(D')\subseteq \Match(\Dhat)$.
\end{itemize}

\noindent{\bf ($\Leftarrow$)}
Assume that there exist $\Dhat\subseteq \Deq c$ and $\Ehat\subseteq \Enq c$ such that all three conditions are satisfied.
We now show that the sets $\Bad c \dhat$ and $\Before c \dhat$ are disjoint.
We show by induction that there is no index $i$ such that $\Badx c {\dhat} i \cap \Before c \dhat\neq \emptyset$.
If $i=0$,  $e\in \Badx c {\dhat} 0$ implies that there does not exist $d$ such that $\Match(d)=e$ and $d_{\bot}\not\prec_c d$.
By the assumption that $\Dhat$ and $\After c \dhat$ are disjoint, we have $d\notin \Dhat$.
But by the assumption that $\Enq c \cap \Before c \dhat \subseteq \Ehat$, we must have $e\in \Ehat$.
This contradicts the assumption that $\Ehat \subseteq \Match(\Dhat)$.

Assume that for all indices less than or equal to $k$, for some $k>0$, the claim holds: $i\leq k$ implies that $\Badx c {\dhat} i$ and $\Before c \dhat$ are disjoint.
Consider the index $k+1$.
Assume that there is $e\in \Badx c {\dhat} {k+1} \cap \Before c \dhat$. 
Then there exists $e_k\in \Badx c {\dhat} k$ such that either $e_k\prec_c e$ or there is $d\in \Deq c$ with $\Match(d)=e$ and $e_k\prec_c d$.
The former case, $e_k\prec_c e$, is not possible since that would imply that $e_k\in \Before c {\dhat}$ and contradict that $\Badx c {\dhat} k \cap \Before c \dhat = \emptyset$.
By the assumption that $\Ehat \cup \Dhat$ is closed under $\prec_c$, $d\in \Dhat$ and $e_k\prec_c d$, we must have $e_k\in \Ehat$.
By the assumption that $\Ehat \subseteq \Match(\Dhat)$ and $e_k\in \Ehat$, there must be $d_k\in \Dhat$ such that $\Match(d_k)=e_k$.
But if $e_k\in \Badx c {\dhat} k$ and $d_k\in \Dhat$, then there must be $e_{k-1}\in \Badx c {\dhat} {k-1}$ such that $e_{k-1}\prec_c e_k$ or $e_{k-1}\prec_c d_k$.
Applying the same arguments as above, we arrive, after $k$ iterations, to the conclusion that there must be some $e_0\in \Badx c {\dhat} 0$ which is also in $\Ehat$.
But by definition, $d_0$ with $\Match(d_0)=e_0$ cannot be in $\Dhat$ (if $d_0$ exists, then $d_0\in \After c {\dhat}$).
This contradicts the assumption that $\Ehat \subseteq \Match(\Dhat)$.
\qed

\begin{defi}
Let $c$ be a complete history with a linearization witness $\Match$.
Call two events $a$ and $a'$ in $c$ {\em overlapping} if neither $a\prec_c a'$ nor $a'\prec_c a$ holds.
We define a relation $\ll_{c,\Match}$ over $\Enq c$. 
For two enqueue events $e_1$ and $e_2$, we have $e_1\ll^1_{c,\Match}e_2$ if $e_1\neq e_2$ and one of the following holds:
\begin{enumerate}
\item $e_1\prec_c e_2$.
\item $e_1$ and $e_2$ are overlapping, there exists $d_1$ such that $\Match(d_1)=e_1$, but there does not exist $d_2$ such that $\Match(d_2)=e_2$.
\item {\sloppy $e_1$ and $e_2$ are overlapping, and there exist $d_1$ and $d_2$ such that $\Match(d_1)=e_1$, $\Match(d_2)=e_2$, and $d_1\prec_c d_2$.}
\item $e_1$ and $e_2$ are overlapping, there exist $d_1$, $d_2$ such that $\Match(d_1)=e_1$, $\Match(d_2)=e_2$, and there exists $d\in \Deq c$ such that $\Val c d=\NULL$, $e_1\notin \Bad c d$ and $e_2\in \Bad c d$.
\end{enumerate}
Let $\ll_{c,\Match}$, called the {\em enq-order}, denote the transitive closure of $\ll^1_{c,\Match}$.
We will drop the subscripts when the history $c$ and its linearization witness either are clear from the context or do not matter.
\end{defi}

\begin{lem}\mylabel{ref:enq-order-total}
Let $c$ be a complete history with linearization witness $\Match$.
Then, the induced enq-order $\ll_{c,\Match}$ is a partial order over $\Enq c$.
\end{lem}
\proof
We have to show that there does not exist a sequence $e_1,\ldots,e_{k+1}$ of enqueue events such that $e_i\ll^1 e_{i+1}$ for $i\in[1,k]$ and $e_{k+1}= e_1$.
The proof is done by induction on $k$, the number of enqueue events in the sequence.
In the base case, we note that $e_1\ll^1 e_1$ is impossible by definition.
Assume that there is no such sequence of length $k$ or less.
Consider the sequence $e_1,\ldots,e_{k+1}$.
For convenience, we will use $d_i$ to denote the dequeue event in $c$ such that $\Match(d_i)=e_i$. 
If no such dequeue event exists for $e_i$, we will say that {\em $d_i$ does not exist}.
We make the following observations about this sequence:

\begin{enumerate}
\item If $d_i$ does not exist, then $d_{i+1}$ cannot exist.
Assume the contrary and that for some $i$, we have $e_i\ll^1 e_{i+1}$, $d_i$ does not exist and $d_{i+1}$ exists.
By the definition of $\ll^1$, $e_i\ll^1 e_{i+1}$ cannot be due to conditions 2-4, because they all require the existence of $d_i$.
Then, we must have $e_i \prec_c e_{i+1}$.
On the other hand, since $\Match$ is a linearization witness for $c$, by condition 2 of ordered mapping, the existence of $d_{i+1}$ implies the existence of $d_i$, which contradicts the assumption that $d_i$ does not exist.
Because the sequence represents a cycle and $\prec_c$ is a partial order, all $d_i$ exist.

\item There cannot be two distinct pairs of events $(e_i,e_{i+1})$ and $(e_j,e_{j+1})$ such that $e_i\prec_c e_{i+1}$ and $e_j\prec_c e_{j+1}$ for some $i<j$.
If there were, then we would have $e_i \prec_c e_{j+1}$ or $e_j \prec_c e_{i+1}$.
If $e_i\prec_c e_{j+1}$, then $e_1\ll^1 e_2\ldots \ll^1 e_i \ll^1 e_{j+1}\ll^1\ldots \ll^1 e_{k+1}$ hold and this sequence does not contain $e_{i+1}$.
If $e_j\prec_c e_{i+1}$, then $e_{i+1}\ll^1 \ldots \ll^1 e_j \ll^1 e_{i+1}$ hold and this sequence does not contain $e_i$.
Thus, both sequences have less than $k+1$ events, which contradict the inductive hypothesis.
\end{enumerate}

\noindent We first show that none of the orderings in the cycle can be due to condition (4); i.e. there is no $i$ such that $e_i\ll^1 e_{i+1}$ because there is some $d\in \Deq c$ such that $e_i\notin \Bad c d$ and $e_{i+1}\in \Bad c d$. 
We assume the contrary and, without loss of generality, assume that $e_1\ll^1 e_2$ is due to condition (4).
Then, there is $d\in \Deq c$ such that $e_1\notin \Bad c d$ and $e_2\in \Bad c d$.
Observe that for all other enqueue events $e_j$ in the sequence, $e_j\notin \Bad c d$ as otherwise, $e_1\ll^1 e_j$ which results in a shorter cycle contradicting the inductive hypothesis.
In particular, $e_3\notin \Bad c d$, but this immediately leads to $e_3\ll^1 e_2$.
This implies that $e_2,e_3,e_2$ is also a cycle.
Thus, if any consecutive events in the cycle are ordered due to condition (4), then $k\leq 2$.
Clearly $e\ll^1 e$ can never hold due to condition (4), leading to the conclusion that if $e_1\ll^1 e_2$ is due to condition (4), then $k=2$.

Now, assume by contradiction that $e_1\ll^1 e_2\ll^1 e_1$ exists and there is $d$ such that $e_1\notin \Bad c d$ and $e_2\in \Bad c d$.
Since $e_1\notin \Bad c d$, $d_1$ exists.
By the first observation above, $d_2$ also exists.
So, $e_2\ll^1 e_1$ cannot be due to condition (2).
We do a case analysis on the possible justifications for $e_2\ll^1 e_1$.
\begin{itemize}
\item Assume that $e_2\prec_c e_1$ (condition (1)). 
By the assumption that $e_2\in \Bad c d$, we have $e_1\in \Bad c d$, which contradicts the assumption that $e_1\notin \Bad c d$.

\item Assume that $e_2$ and $e_1$ are overlapping and $d_2\prec_c d_1$ (condition (3)).
Because $e_2\in \Bad c d$, either $d\prec_c d_2$ or $d\prec_c e_2$ or there is an enqueue event $e'\in \Bad c d$ such that either $e'\prec_c e_2$ or $e'\prec_c d_2$ holds.
If $d\prec_c d_2$ holds, then by transitivity $d\prec_c d_1$ also holds.
If $d\prec_c e_2$ holds, then because $d_2\prec_c e_2$ cannot hold ($\Match$ is ordered), $d\prec_c d_1$ must hold.
If $e'\prec_c e_2$ holds, then because $d_2\not\prec_c e_2$ holds (due to $\Match$ being ordered) we must have $e'\prec_c d_1$.
Finally, if $e'\prec_c d_2$ holds, then by transitivity $e'\prec_c d_1$ also holds.
All four cases contradict the assumption that $e_1\notin \Bad c d$.

\item Assume that there exists $d'\in \Deq c$ such that $\Match(d')=\bot$, $e_2\notin \Bad c {d'}$ and $e_1\in \Bad c {d'}$ (condition (4)).
We do a case analysis on the possible justifications of $e_1\in \Bad c {d'}$ and $e_2\in \Bad c d$ holding:
\begin{itemize}

\item $d'\prec_c e_1$, and $d\prec_c e_2$. 
Then either $d'\prec_c e_2$ or $d\prec_c e_1$ holds.

\item $d'\prec_c e_1$, and $d\prec_c d_2$.
Then either $d'\prec_c d_2$ or $d\prec_c e_1$ holds.

\item $d'\prec_c e_1$, and there is $e_{b,d}\in \Bad c d$ such that $e_{b,d}\prec_c e_2$.
Then either $d'\prec_c e_2$ or $e_{b,d}\prec_c e_1$ holds.

\item $d'\prec_c e_1$, and there is $e_{b,d}\in \Bad c d$ such that $e_{b,d}\prec_c d_2$.
Then either $d'\prec_c d_2$ or $e_{b,d}\prec_c e_1$ holds.

\item $d'\prec_c d_1$, and $d\prec_c e_2$.
Then either $d'\prec_c e_2$ or $d\prec_c d_1$ holds.

\item $d'\prec_c d_1$, and $d\prec_c d_2$.
Then either $d'\prec_c d_2$ or $d\prec_c d_1$ holds.

\item $d'\prec_c d_1$, and there is $e_{b,d}\in \Bad c d$ such that $e_{b,d}\prec_c e_2$.
Then either $d'\prec_c e_2$ or $e_{b,d}\prec_c d_1$ holds.

\item $d'\prec_c d_1$, and there is $e_{b,d}\in \Bad c d$ such that $e_{b,d}\prec_c d_2$.
Then either $d'\prec_c d_2$ or $e_{b,d}\prec_c d_1$ holds.

\item There is $e_{b,d'}\in \Bad c {d'}$ such that $e_{b,d'}\prec_c e_1$, and $d\prec_c e_2$.
Then either $e_{b,d'}\prec_c e_2$ or $d\prec_c e_1$ holds.

\item There is $e_{b,d'}\in \Bad c {d'}$ such that $e_{b,d'}\prec_c e_1$, and $d\prec_c d_2$.
Then either $e_{b,d'}\prec_c d_2$ or $d\prec_c e_1$ holds.

\item There is $e_{b,d'}\in \Bad c {d'}$ such that $e_{b,d'}\prec_c e_1$, and there is $e_{b,d}\in \Bad c d$ such that $e_{b,d}\prec_c e_2$.
Then either $e_{b,d'}\prec_c e_2$ or $e_{b,d}\prec_c e_1$ holds.

\item There is $e_{b,d'}\in \Bad c {d'}$ such that $e_{b,d'}\prec_c e_1$, and there is $e_{b,d}\in \Bad c d$ such that $e_{b,d}\prec_c d_2$.
Then either $e_{b,d'}\prec_c d_2$ or $e_{b,d}\prec_c e_1$.

\item There is $e_{b,d'}\in \Bad c {d'}$ such that $e_{b,d'}\prec_c d_1$, and $d\prec_c e_2$.
Then either $e_{b,d'}\prec_c e_2$ or $d\prec_c d_1$ holds.

\item There is $e_{b,d'}\in \Bad c {d'}$ such that $e_{b,d'}\prec_c d_1$, and $d\prec_c d_2$.
Then either $e_{b,d'}\prec_c d_2$ or $d\prec_c d_1$ holds.

\item There is $e_{b,d'}\in \Bad c {d'}$ such that $e_{b,d'}\prec_c d_1$, and there is $e_{b,d}\in \Bad c d$ such that $e_{b,d}\prec_c e_2$.
Then either $e_{b,d'}\prec_c e_2$ or $e_{b,d}\prec_c d_1$ holds.

\item There is $e_{b,d'}\in \Bad c {d'}$ such that $e_{b,d'}\prec_c d_1$, and there is $e_{b,d}\in \Bad c d$ such that $e_{b,d}\prec_c d_2$.
Then either $e_{b,d'}\prec_c d_2$ or $e_{b,d}\prec_c d_1$.

\end{itemize}
In all cases the former implication contradicts $e_2\notin \Bad c {d'}$ and the latter implication contradicts $e_1\notin \Bad c d$.

\end{itemize}
Thus, if $e_1,e_2,\ldots,e_{k+1}$ is a cycle in $\ll^1$, none of the pairwise orderings can be due to condition (4).

Now consider the case where all consecutive events are overlapping; that is, $e_i$ and $e_{i+1}$ are overlapping for all $i\in[1,k]$.
Then, by the definition of $\ll^1$ and the first observation, we must have $d_i\prec_c d_{i+1}$.
But this would imply by the transitivity of $\prec_c$ that $d_1\prec_c d_{k+1}=d_1$ which is impossible due to $\prec_c$ being a partial order.

So, there must be exactly one pair $e_i$ and $e_{i+1}$ of events ordered by $\prec_c$.
Without loss of generality assume that $e_1\prec_c e_2$.
By the second observation, $e_2$ is overlapping with all $e_{i+1}$ for $i\in[2,k]$.
In particular, $e_2$ and $e_{k+1}=e_1$ must be overlapping.
That contradicts the assumption that $e_1\prec_c e_2$. 
Thus no sequence of length $k+1$ can have a cycle in the $\ll^1$ relation.
\qed

The main result of this section is stated below.
\begin{thm}\mylabel{thm:witness}
A set of histories $C$ is linearizable with respect to queue iff every $c \in C$ has a completion
$\compl c \in \Compl c$ that has a linearization witness.
\end{thm}
\newpage

\proof\hfill

\noindent{\bf ($\Rightarrow$)} If $c \in C$ is linearizable with respect to queue, then there is a linearization $s$ of $c$ which is a legal queue behavior.
By Theorem~\ref{thm:equiv-legal-seqwitness}, $s$ has a sequential witness $\mus$.
The mapping $\mus$ satisfies the conditions of a linearization witness since all $\prec_c$ orderings are preserved in $s$.
In particular, $\mus$ is safe because conditions (i) to (iii) of sequential witness imply conditions (1) to (3) of safe mapping.
It is ordered because
\begin{itemize}
\item By condition (iv) of sequential witness, $\mus(d)=e$ implies $e\prec_s d$ and definition of linearizability implies that $d\not\prec_c e$, which is condition (1) of ordered mapping,
\item Assume that there exist $d',e',e$ such that $e'=\mus(d')$ and $e\prec_c e'$.
Then by definition of linearization, $e\prec_s e'$.
By condition (v) of sequential witness, $d=\mus^{-1}(e)$ exists and $d\prec_s d'$.
By definition of linearization, this in turn implies that $d'\not\prec_c d$, which is condition of (2) of ordered mapping.
\end{itemize}
Assume $d=\deq(\NULL)\in \Deq c$.
Define the sets $D_d=\{d'\in \Deq s \mid d'\prec_s d\}$, $E_d=\{e\in \Enq s \mid e\prec_s d\}$.
Observe that $(D_d\cup E_d)\cap \After d c=\emptyset$ because for any $a\in \After d c$, by definition we have $d\prec_c a$, which implies $d\prec_s a$, which in turn implies $a\notin D_d\cup E_d$.
Assume there is $e\in \Before d c \cap \Enq c$.
Then by definition of linearization, $e\prec_s d$.
By construction, $\Before d c\cap \Enq c\subseteq E_d$.
Let $i$ denote the position of $d$ in $s$; i.e. $\seqx s i=d$.
Because $s$ is legal, it has an obs-equivalent canonical behavior, $s'$.
By Lemma~\ref{lem:seqwit-canonical} $\seqx {s'} i=d$.
By definition of canonical behavior, each enqueue event in $\seqx {s'} {1:i-1}$ has a matching dequeue event in $\seqx {s'} {1:i-1}$.
Since $s$ and $s'$ are obs-equivalent, then each enqueue event in $\seqx {s} {1:i-1}$ has a matching dequeue event in $\seqx {s} {1:i-1}$.
Thus, $E_d\subseteq \Match(D_d)$, the inclusion being proper in case $D_d$ contains a {\NULL}-dequeue event (distinct from $d$ since $d\notin D_d$).
Thus, $\Before d c\cap \Enq c\subseteq E_d\subseteq \Match(D_d)$.
Since all conditions of linearization witness per Lem.~\ref{lem:alt-lin-witness} are satisfied for $D_d$ and $E_d$, $\mus$ is a linearization witness.

\noindent{\bf ($\Leftarrow$)}
Let $c$ be a complete history with a linearization witness $\Match$.
Let $<$ denote a total order extension of $\ll$.
That is, $<$ is a total order over $\Enq c$ such that whenever $e\ll e'$, we have $e<e'$.
Let $e^*$ denote the $<$-maximal enqueue event over $<$.
That is, for any $e\in \Enq c$, we have $e<e^*$ whenever $e\neq e^*$.

In order to prove the if-direction ($\Leftarrow$), we will make use of $<$ to construct a sequence $s$ with sequential witness $\mu$.
We actually prove a stronger property, which also requires that if $e<e'$ in $c$ then $e\prec_s e'$.
By Theorem~\ref{thm:equiv-legal-seqwitness}, the result follows.

The construction is given by induction on the number of (completed) events in $c$.
In the base case, there are no events and $\varepsilon$ with empty mapping is the desired sequence.
Assume that the claim holds for all complete concurrent histories with $k$ events or less.
Let $c$ be a complete concurrent history with $k+1$ events and $\Match$ be a linearization witness for $c$.
We first choose an event.

Call event $a\in \Enq c\cup \Deq c$ {\em maximal} (relative to $<$), if there is no event $a'$ such that $a\prec_c a'$ and one of the following holds:
\begin{enumerate}
\item $a=e^*\in \Enq c$, there is no $d^*$ such that $\Match(d^*)=e^*$.
\item $a=d\in \Deq c$ with $\Match(d)\neq\bot$, there is no $d'$ such that $\Match(d)<\Match(d')$.
\item $a=\dhat\in \Deq c$ with $\Match(\dhat)=\bot$, and $\Bad c {\dhat}=\emptyset$.
\end{enumerate}
Let $c$ be a non-empty complete history and $\Match$ be its linearization witness.
We first show that there is at least one event in $c$ that is maximal relative to $<$.
First, observe that if $\Enq c=\emptyset$, then any $d\in \Deq c$ must return {\NULL}; otherwise, $\Match$ cannot be safe.
Then, any $d$ such that no $d'\in \Deq c$ with $d\prec_c d'$ exists is maximal.
Since $\prec_c$ is a partial-order, such $d$ must exist.
If conversely we assume that $\Enq c\neq\emptyset$ and $\Deq c$ is empty, then $e^*$ is maximal.

Assume that $\Enq c$ and $\Deq c$ are non-empty.
If $e^*$ is not maximal, it must be because there is $d^*\in \Deq c$ such that $\Match(d^*)=e^*$.
Then, by definition of $<$ and the assumption that $e^*$ is $<$-maximal, there cannot be $d'\in \Deq c$ such that $d^*\prec_c d'$ if $\Match(d')\neq\bot$.
So, $d^*$ is not maximal only if there is $\dhat\in \Deq c$, $\Match(\dhat)=\bot$ and $d^*\prec_c \dhat$.
Furthermore, the definition of $\ll^1$, that $e^*$ is $<$-maximal and $d^*$ exists imply that for all $e'\in \Enq c$, there is  $d'\in \Deq c$ such that $\Match(d')=e'$.
In particular, this means that for $\dhat\in \Deq c$ such that no $d'\in \Deq c$ with $\dhat\prec_c d'$ exists and $d^*\prec_c \dhat$ with $\Bad c {\dhat}=\emptyset$, setting $\dhat$ as a maximal element.
Thus, the set of maximal events in any non-empty history is non-empty. 

Let $A$ denote the set of maximal elements relative to $<$.
If $A$ contains a dequeue event $d$ such that $\Match(d)=\bot$, then we choose $d$.
Otherwise, if $A$ contains a dequeue event $d^*$ such that $\Match(d^*)\neq\bot$, then we choose $d^*$.
If neither condition holds, we choose $e^*$.

We now show that if $c$ is a non-empty history with linearization witness $\Match$, the history $c'$ obtained by removing the chosen event from $c$ has $\Match'$, which is $\Match$ restricted to the remaining events in $c'$, as a linearization witness.
Before we do a case analysis on the type of the chosen event, we make two observations. 
If $c'$ is obtained from $c$ by removing an event $a$ and a mapping is safe for $c$, then it is also safe for $c'$ when restricted to the $\Deq {c'}$.
Second, removing $a$ from $c$ does not change the relative ordering among the remaining events.
So $b\prec_c d$ holds iff $b\prec_{c'} d$ holds.
In particular, if $a\in \Deq c$ and a mapping is ordered for $c$, then it is ordered for $c'$.

We have three cases to consider for the chosen event:
\begin{itemize} 
\item The chosen event is $\dhat$ with $\Match(\dhat)=\bot$.
Let $d'\in \Deq c$ be such that $\Match(d')=\bot$.
Since $\Match(\dhat)=\bot$, after removing $\dhat$ we have $\Enq {c'}=\Enq c$ and thus $\Bad c {d'}=\Bad {c'} {d'}$.
Additionally, $\Before c {d'}$ is the same as $\Before {c'} {d'}$ when both are restricted to $\Enq c=\Enq {c'}$.
Then, we have
\begin{align*}
&\Before c {d'} \cap \Bad c {d'} &&\\
&\quad=\Before c {d'} \cap \Bad c {d'} \cap \Enq c &&\textnormal{\small [$\Bad c {d'}\subseteq \Enq c$]}\\
&\quad= \Before c {d'} \cap \Bad c {d'} \cap \Enq {c'} &&\textnormal{\small[$\Enq c=\Enq {c'}$]}\\
&\quad= \Before c {d'} \cap \Bad {c'} {d'} \cap \Enq {c'} &&\textnormal{[\small$\dhat\notin \Bad c {d'}$]}\\
&\quad= \Before {c'} {d'} \cap \Bad {c'} {d'} \cap \Enq {c'} &&\textnormal{[\small$\Before c {d'}\cap \Enq c=\Before {c'} {d'}\cap \Enq {c'}$]}\\
&\quad= \Before {c'} {d'} \cap \Bad {c'} {d'} &&\textnormal{\small[$\Bad {c'} {d'}\subseteq \Enq {c'}$]}
\end{align*}
establishing that $\Before {c'} {d'}\cap \Bad {c'} {d'}=\emptyset$.
Thus, $\Match'$ is a linearization witness for $c'$.

\item The chosen event is $d^*\in\Deq c$. 
Observe that $\Match(d^*)$ is the $<$-maximal enqueue event $e^*$ relative to $<$.
By the second observation above, $\Match'$ is ordered for $c'$.
We have to show that for any $\dhat\in \Deq c$, $\Match'(\dhat)=\bot$ is justified; that is, $\Bad {c'} {\dhat}\cap \Before {c'} {\dhat}=\emptyset$.
By the assumption that $\Match$ is a linearization witness for $c$, we have $\Bad c {\dhat} \cap \Before c {\dhat}=\emptyset$.
If $\dhat \prec_c e^*$, then $e^*\in \Bad {c'} {\dhat}$ by definition.
If $\dhat \prec_c d^*$, then $e^*\in \Badx c {\dhat} 0$ and $e^*\in \Badx {c'} {\dhat} 0$, so $\Bad {c'} {\dhat}=\Bad c {\dhat}$.

Then, the interesting case is when $e^*\notin \Bad c {d}$. 
First observe that $\Bad c {d}\neq\emptyset$ iff $e^*\in \Bad c {d}$.
For the only-if ($\Rightarrow$) direction, assume that there is some $e'\in \Bad c {d}$.
By the definition of $\ll^1$, if $e^*\notin \Bad c {d}$ then $e^*\ll^1 e'$ contradicting the $<$-maximality of $e^*$.
The if ($\Leftarrow$) direction is trivial.
This implies that $\Bad c {d}=\emptyset$ because $e^*\notin \Bad c {\dhat}$.
If there are several such \NULL-returning dequeues, choose $\dhat$ such that for any $a\in \Deq c$ with $\dhat\prec_c a$ implies $\Match(d)\neq\bot$.
Intuitively, $\dhat$ is the $<$-maximal among dequeue events returning \NULL.

Now since $d^*$ was chosen, we know that there must be at least one $a$ such that $\dhat\prec_c a$, since otherwise $\dhat$ would have been chosen.
By the assumption about $\dhat$, $a\notin \Deq c$ with $\Match(a)=\bot$.
If $a=e\in \Enq c$, then $e\in \Bad c {\dhat}$ contradicting the assumption that $\Bad c {\dhat}=\emptyset$.
So $a=d\in \Deq c$ with $\Match(d)\neq \bot$.
But then $\Match(a)$, which must exist because $\Match$ is safe, is in $\Bad c {\dhat}$, again contradicting the assumption that $\Bad c {\dhat}=\emptyset$.
So, by contradiction we conclude that there is no such $\dhat$ for which $\Bad c {\dhat}=\emptyset$ and $\Bad {c'} {\dhat}\neq\emptyset$ hold.

\item The chosen event is $e^*\in \Enq c$.
By the assumption about the chosen event, $d^*$ does not exist, so $\Deq c=\Deq {c'}$ and $\Match'$ is safe because $\Match$ is safe.
Because $d^*$ does not exist, if $\dhat$ is such that $\Match(\dhat)=\bot$, then $e^*\in \Bad c {\dhat}$. 
Then, for every such $\dhat$, $\Bad {c'} {\dhat}\subseteq \Bad c {\dhat}$, which means that $\Bad c {\dhat} \cap \Before c {\dhat}=\emptyset$ implies $\Bad c {\dhat} \cap \Before {c'} {\dhat}$. 
So, $\Match'$ is a linearization witness for $c'$.

\end{itemize}

\noindent Now, we know that $\Match'$ is a linearization witness for $c'$ which has exactly $k$ events.
By the inductive hypothesis, $c'$ is linearizable with respect to queue. 
That is, there is a linearization $s'$ of $c'$ which is a legal queue behavior.
By Theorem~\ref{thm:equiv-legal-seqwitness}, $s'$ has a sequential witness $\mu'$.
We claim that $s=s'\cdot a$, where $a$ is the chosen element in $c$ as described above, is a legal queue behavior.
Additionally, we will also show that for any two enqueue events $e$ and $e'$ both in $\Enq c$, $e<e'$ implies $e\prec_s e'$.

Assume that the chosen element was $a=\dhat\in \Deq c$ such that $\Match(\dhat)=\bot$.
We set $\mu=\mu'[a\mapsto \bot]$.
Observe that by the assumption that $\dhat$ is a chosen element, we must have $\Bad c {\dhat}=\emptyset$.
This implies that for all $e\in \Enq c$, there is $d\in \Deq c$ such that $\Match(d)=e$; as otherwise, $e$ would be in $\Bad c {\dhat}$.
Since all events of $c'$ are the same as the events of $s'$, the sets $\{e\in \Enq c \mid e\prec_s \dhat\}=\Enq c$ and $\{d\in \Deq c \mid d\prec_s \dhat \wedge \mu(d)\neq\bot\}$ have the same cardinality.
These along with the inductive hypothesis that $\mu'$ is a sequential witness for $s'$ imply that all six conditions of a sequential witness are satisfied for $\mu$ and $s$.
Because the relative ordering of events in $\Enq c$ in $s'$ remains the same in $s$, $e\prec_{s'} e'$ implies $e\prec_s e'$, and by induction hypothesis this can happen only when $e<e'$.

Assume that the chosen element was $a=d^*\in \Deq c$ such that $\Match(d^*)=e^*$. 
We set $\mu=\mu'[a\mapsto e^*]$.
Because $\Match$ was safe for $c$, $e^*$ exists and $\mu$ is well-defined.
By the inductive hypothesis, $e^*$ is in $s'$ and hence $e^*\prec_s d^*$.
Again by the inductive hypothesis, for any $e\in \Enq c$, we have $e\prec_s e^*$.
Since $d^*$ is the last event in $s$, no event can follow $d^*$ in $s$.
In particular, there is no $d'\in \Deq c$ such that $d^*\prec_s d'$. 
These along with the inductive hypothesis imply that $\mu$ is a sequential witness for $s$.
Similar to the previous case, $e\prec_{s'} e'$ implies $e\prec_s e'$ and by inductive hypothesis this can happen only when $e<e'$.

Assume that the chosen element was $a=e^*$.
We take $\mu=\mu'$.
Because $e^*$ is chosen, $d^*$ does not exist in $c$.
Furthermore, since $e^*$ is the last event in $s$, no other event can follow $e^*$ in $s$.
These observations along with the inductive hypothesis imply that $\mu$ is a sequential witness for $s$.
Observe also that $e^*$, being the last element in $s$, also satisfies the condition that it should not precede any other enqueue event in $s$, satisfying the condition that $e<e'$ implies $e\prec_s e'$.
\qed

% !TEX root = main.tex
\subsection*{Necessary and Sufficient Conditions for Complete Histories}

We now focus on complete histories, namely ones with no pending events. 
We observe that whether a history is not linearizable can always be determined by examining the dequeued values. 
Let $c$ be a complete history.
In order to simplify the technical presentation we assume that each value is enqueued at most once.\footnote{In case there are multiple occurring values, this is akin to guessing the mapping {\Match}; it is enough that at least one guess satisfies the criteria (absence of violations).}
The possible violations in $c$ are: 
\begin{description}
\item[(\VFresh)] A dequeue event returns a value not previously inserted by any enqueue event.
Formally, there exists a value $x\neq\NULL$ such that $\deq(x)\in \Deq c$ and either $\enq(x) \notin \Enq c$ or $\deq(x) \prec_c \enq(x)$.
\item[(\VRepet)] Two dequeue events return the value inserted by the same enqueue event.
Formally, there exist two dequeue events $d,d'\in \Deq c$ such that $\Val c d=\Val c {d'}\neq\NULL$.
\item[(\VOrd)]
Two values are enqueued in a certain order, and a dequeue returns the later value before any dequeue of the earlier value starts.
Formally, there exist values $x$, $y$ such that $\enq(y) \prec_c \enq(x)$, $\deq(x)\in \Deq c$, and either $\deq(y)\notin \Deq c$ or $\deq(x) \prec_c \deq(y)$.
\item[(\VWit)] A dequeue event returning \NULL\ even though the queue is never logically empty during the execution of the dequeue event.
Formally, let $c=c_0\cdot \deq_i(\NULL) \cdot c_d\cdot \deq_r(\NULL) \cdot c_3$, where $c_0,c_d,c_3$ represent subsequences of $c$.
Then for any choice of $c_1$ and $c_2$ such that $c_d=c_1\cdot c_2$, there exists an $\enq(x)\in \Enq c$ completed in $c_0\cdot \deq_i(\NULL) \cdot c_1$ and $\deq_i(x)$ does not occur in $c_0\cdot \deq_i(\NULL) \cdot c_1$.
\end{description}\smallskip

\noindent We have the following result which ties the above violation types to linearizable queues.

\begin{prop}\mylabel{prop:compl-viol}
A complete history $c$ is linearizable with respect to queue iff it has none of the \VFresh, \VRepet, \VOrd, \VWit\ violations.
\end{prop}

\begin{proof}\hfill

\noindent{\bf($\Rightarrow$)} If $c$ is linearizable with respect to queue, then by Theorem~\ref{thm:witness}, $\compl{c}=c$ has a linearization witness $\Match$.
We show by contradiction that none of the four violations can happen in $c$.

\begin{itemize}
\item Assume that $c$ has \VFresh.
Then there exists a dequeue event $d\in \Deq c$ such that $\Val c d\neq \NULL$ and either $e=\Match(d)$ does not exist or $d\prec_c e=\Match(d)$.
That $e=\Match(d)$ does not exist is impossible because by the second condition of safe mapping, $\Match(d)\neq\bot$ and by the first condition of safe mapping $\Match(d)\in \Enq c$.
That $d\prec_c e=\Match$ holds is impossible because by the first condition of safe mapping, $d\not\prec_c \Match(d)=e$.

\item Assume that $c$ has \VRepet.
Then there exist $d,d'\in \Deq c$ with $\Val c d=\Val c {d'}\neq \bot$.
This is impossible by the third condition of safe mapping.

\item Assume that $c$ has \VOrd.
Then there exist $e,e'\in \Enq c$, $d'\in \Deq c$ such that $e\prec_c e'=\Match(d')$ and either $d\in \Deq c$ such that $\Match(d)=e$ does not exist or such a $d$ exists and $d'\prec_c d$.
Both possibilities contradict the second condition of ordered mapping.

\item Assume that $c$ has \VWit.
Then $c$ is of the form $c_0\cdot inv(\dhat)\cdot c_d\cdot res(\dhat)\cdot c_3$ such that $\Match (\dhat)=\bot$, and for every possible partitioning of $c_d=c_1\cdot c_2$, there is an enqueue event $e\in \Enq {c_p}$ with $c_p=c_0\cdot inv(\dhat)\cdot c_1$ such that $e$ is completed in $c_p$ and there is no dequeue event $d$, pending or completed, in $c_p$ such that $\Match(d)=e$.
First, observe that by choosing $c_2=\varepsilon$ (resulting in $c_p=c_0\cdot inv(\dhat)\cdot c_d$), we conclude that there is at least one enqueue event $e_0\in \Enq {c_p}$ whose matching dequeue event $d_0$ is not in $\Enq {c_p}$; that is, $\dhat\prec_c d_0$ if $d_0\in \Deq c$.
This implies that $e_0\in \Bad c {\dhat}$.
Because $\Match$ is a linearization witness for $c$, we must have $\Bad c {\dhat}\cap \Before c {\dhat}=\emptyset$.
In other words, all enqueue events $e\in \Bad c {\dhat}$ must not belong to $\Before c {\dhat}$.
This implies that if $e\in \Bad c {\dhat}$ then $res(e)$ must happen after $inv(\dhat)$.
Let $e\in \Bad c {\dhat}$ be chosen such that for any other $e'\in \Bad c {\dhat}$, $res(e)$ occurs before $res(e')$ in $c$.
Let $c_d=c_1\cdot c_2$ with $c_2=res(e)\cdot c_2'$. 
By the assumption that there is a $\VWit$ violation for $\dhat$, there must be an enqueue event $e'$ in $c_p=c_0\cdot inv(\dhat)\cdot c_1$ such that if there is $d'\in \Deq c$ with $\Match(d')=e'$, then $d'$ is neither completed nor pending in $c_p$.
This implies that $inv(d')$ if it exists must occur after $res(e)$.
Because $e$ is not completed in $c_p$ (it is completed in $c_p\cdot res(e)$), $e'\neq e$.
These two facts imply that either $d'\notin \Deq c$ or if $d'\in \Deq c$ then $e\prec_c d'$ holds.
But this implies that $e'\in \Bad c {\dhat}$.
This contradicts the assumption that $res(e)$ is the first enqueue event in $\Bad c {\dhat}$ to complete in $c$.
Such an $e$ does not exist implies that there is at least one enqueue event $e_b$ in $\Bad c {\dhat}$ which is completed in $c_0$, which implies that $e_b\in \Before c {\dhat}$. 
Finally, this contradicts the assumption that $\Bad c {\dhat}$ and $\Before c {\dhat}$ are disjoint.
\end{itemize}

\noindent{\bf($\Leftarrow$)} 
Assume that there exists a complete history $c$ in which none of the violations happen.
We will show that the mapping that pairs events enqueueing and dequeueing the same value is a linearization witness for $c$.

Let $D_v=\{\deq(x)\in {\Deq c} \mid x\neq\NULL\}$ denote the set of all {non-\NULL} returning dequeue events of $c$.
Similarly, let $D_n=\Deq{c}\setminus D_v$ denote the set of all {\NULL} returning dequeue events of $c$.
Let $M_v$ be the mapping from $D_v$ to $\Enq c$ such that $M_v(d)=e$ iff $\Val c d = \Val c e$. 
Let $M_n$ be such that all $d\in D_n$ are mapped to $\bot$.
We claim that $\Match$ defined as 
\[
\Match(d) \defeq 
	\begin{cases}
		M_v(d) & \text{if } d\in D_v \\
		\bot & \text{if } d\in D_n
	\end{cases}
\]
is a linearization witness for $c$.

First, observe that $M_v$ is a total mapping because $c$ does not have {\VFresh}.
Furthermore, because $c$ does not contain {\VRepet}, $\Match$ is a safe mapping by construction.
$\Match$ satisfies the first condition of an ordered mapping because $c$ does not have {\VFresh}.
$\Match$ satisfies the second condition of an ordered mapping because $c$ does not have {\VOrd}.
Thus, $\Match$ is also an ordered mapping.

Let $\dhat\in D_n$ be a {\NULL}-returning dequeue event in $c$.
We have to show that $\Before c {\dhat}$ and $\Bad c {\dhat}$ are disjoint.
Because $c$ has no {\VWit} violation, there must be a prefix $c_p=c_0\cdot inv(\dhat)\cdot c_1$ of $c$ such that if $e$ is an enqueue event is completed in $c_p$ then its matching dequeue event $d$ (i.e. $\Match(d)=e$) is either pending or completed in $c_p$.
In other words, if $res(e)$ occurs in $c_p$, then so does $inv(d)$.
Let $e_j\in \Bad c {\dhat}$ be such that $e_j\in \Badx c {\dhat} j$, for any $e_k\in \Bad c {\dhat}$ we have $j\leq k$, and $e_j$ is completed in $c_p$.
If there is no such $e_j$, that is, if $\Bad c {\dhat}$ is empty, then we are done.
Otherwise, observe that $j\neq 0$ because by the absence of {\VWit}, there is $d_j$ such that $\Match(d_j)=e_j$ and $\dhat\not\prec d_j$; in particular, $inv(d_j)$ occurs in $c_p$.
But $j>0$ implies that there is $e_{j-1}\in \Badx c {\dhat} {j-1}$ such that either $e_{j-1}\prec_c e_j$ or $e_{j-1}\prec_c d_j$.
Both cases imply that $e_{j-1}$ must be completed in $c_p$, contradicting the assumption that $j$ was minimal.
Thus, there are no enqueue events in $\Bad c {\dhat}$ which are completed in $c_p$.
Since $\Before c {\dhat}$ is contained in the set of completed events of $c_p$, we conclude that $\Before c {\dhat}$ and $\Bad c {\dhat}$ are disjoint.

This concludes the proof that $\Match$ is a linearization witness for $c$.
\end{proof}

We remark that none of the violations  mentions the possibility of an
element inserted by an enqueue being lost forever.  This is intentional, as
such histories are ruled out by the following proposition.

\begin{prop}
Given an infinite sequence of complete histories $c_1,c_2,\ldots$ 
not containing any of the violations above, 
where for every $i$, $c_i$ is a prefix of
$c_{i+1}$, and the number of dequeue events in $c_i$ is less than that of $c_{i+1}$,
if $c_1$ contains an enqueue event $\enq(x)$,
then exists some $c_j$ containing $\deq(x)$.
\end{prop}

\begin{proof}
We prove this by contradiction. 
If there is no $\deq(x)$ event, then $\enq(x)$ is always in the queue, 
and so, from the absence of \VWit\ violations, none of the dequeue events
following $\enq(x)$ can return \NULL.
Also, since dequeue events cannot return values that were not previously
enqueued {\VFresh} and cannot return the same value multiple times {\VRepet},
and since the number of dequeue events is increasing, then there must also be
new enqueue events.
However, only finitely many of those are not preceded by $\enq(x)$ which
completes in $c_1$. 
This means that eventually one dequeue event has to return an element inserted
by $\enq(y)$ such that $\enq(x)\prec_{c_j}\enq(y)$, which is \VOrd.
\end{proof}

For checking purposes, we find it useful to re-state the third violation as the following
equivalent proof obligation.
\begin{description}
\item[(\POrd)]
For any enqueue events $e_1$ and $e_2$ with $e_1\prec_c e_2$ and \mbox{$\Val{c}{e_1}\neq\Val{c}{e_2}$},
a dequeue event $d_2$ cannot return $\Val{c}{e_2}$ 
if $\Val{c}{e_1}$ is not removed in $c$ or is removed by $d_1$ with $d_2\prec_c d_1$.
\end{description}
Thus, to check this property, it suffices to come up with an overapproximation of
all those executions satisfying the premise of \POrd, and prove that such executions
cannot end with a dequeue event (in the sense that no other method is preceded 
by that dequeue event) returning the value of $e_2$.

\subsection*{Necessary and Sufficient Conditions for Purely-Blocking Queues}

There is a subtle complication in the statement of Theorem~\ref{thm:witness}.
The witness mapping is chosen relative to some completion of the concurrent history under
consideration. 
However, because implementations may become blocked, such completions may actually never be reached. 
This means that one cannot reason about the correctness of a queue implementation by considering only the reachable states of the implementation.
What we would ideally like to do is to claim that if the implementation violates linearizability, then there is a finite complete induced history of the implementation which has no witness. 
In other words, if the implementation contains an incomplete execution trace whose induced (incomplete) history has no witness, then that execution trace is the prefix of a complete execution trace of the implementation.

Let $C$ be the set of all induced histories of a library implementation.
We call the library implementation \emph{completable} iff for every history $c
\in C$, we have $\Compl c \cap C \neq \emptyset$.
For completable implementations, it suffices to consider only complete execution traces.

\begin{thm}\label{thm:no-violations}
A completable queue implementation is linearizable 
iff all its complete histories have none of the \VFresh, \VRepet, \VOrd\ and \VWit\ violations.
\end{thm}

\begin{proof}\hfill

\noindent{\bf($\Rightarrow$)} If some complete history has a violation, by Prop.~\ref{prop:compl-viol}, it has no linearization, contradicting the assumption that the implementation is linearizable.

\noindent{\bf($\Leftarrow$)}
Consider an arbitrary induced history $c$ of the implementation. 
As the implementation is completable, 
there exists a completion $\compl c \in \Compl c$ that is a valid induced history of the implementation.
From our assumptions, $\compl c$ cannot have a violation, and so by
Prop.~\ref{prop:compl-viol}, $\compl c$ has a linearization, and therefore so does $c$.
\end{proof}

Since it may not be obvious how to easily prove that an implementation is
completable, we introduce the stronger notion of purely-blocking
implementations, that is straightforward to check.
We say that an implementation is \emph{purely-blocking} when at any reachable state,
any pending method, if run in isolation will terminate or its entire execution does
not modify the global state.
Formally, let $\tau=\tau_0\cdot (t:enter(m))\cdot \tau_1$ be an execution trace of the implementation in which $m$ executed by $t$ is pending, i.e. $(t:exit(m))$ does not occur in $\tau_1$.
The pending method $m$ is called {\em pure after $\tau$} if for any sequence $\tau_e$ in which no action of $m$ by $t$ occurs and any sequence $\tau_m$ in which only actions of $m$ by $t$ occur, $\tau\cdot \tau_e$ is an execution trace of the implementation iff $\tau\cdot \tau_m\cdot \tau_e$ is an execution trace of the implementation.
The execution trace $\tau$ is called {\em obstruction-free for $m$} if there is another execution trace $\tau'=\tau\cdot \tau_2\cdot (t:exit(m))$ of the implementation such that all actions in $\tau_2$ belong to $m$ executed by $t$.
Then, the implementation is purely-blocking if for each execution trace $\tau$ of the implementation and pending method $m$ in $\tau$, either $\tau$ is obstruction-free for $m$ or $m$ is pure after $\tau$.

\begin{prop}\label{prop:pb}
Every purely-blocking implementation is completable.
\end{prop}

\begin{proof}
Let $\tau$ be an execution trace of a purely-blocking implementation.
We fix a total order of pending methods, and consider them in that order. 
For a pending method $m$ executed by $t$, if running it in isolation terminates, then extend $\tau$ only with actions executed by $t$ until $(t:exit(m))$ occurs.
Otherwise, the execution of $m$ does not modify any global state and so all actions executed by $t$ beginning with the last occurrence of $(t:enter(m))$ can be removed from the execution trace without affecting its realizability.
\end{proof}

We remark that our new notion of purely-blocking is a strictly weaker
requirement than the standard non-blocking notions:
\emph{obstruction-freedom}, which requires all pending methods to terminate when run in isolation, 
as well as the stronger notions of lock-freedom and wait-freedom.
(See~\cite{HS2008} for an in depth exposition of these three notions.)

% !TEX root = main.tex

\section{Manually Verifying the Herlihy-Wing Queue}
\label{sec:herlihy-wing}

Let us return to the HW queue presented in \S\ref{sec:introduction} and prove
its correctness manually following our aspect-oriented approach.

First, observe that HW queue is purely-blocking: $\enq()$ always terminates,
and $\deq()$ can update the global state only by reading $x\neq\NULL$ at $E_2$,
in which case it immediately terminates.
So from Prop.~\ref{prop:pb} and Theorem~\ref{thm:no-violations}, it suffices
to show that it does not have any of the four violations.
The last one, \VWit, is trivial as the HW $\deq()$ never returns $\NULL$.
So, we are left with three violations
whose absence we have to verify: \VFresh, \VRepet, and \VOrd. 

Intuitively, there are no \VFresh\ violations because $\deq()$ can return only a value that
has been stored inside the $q.items$ array.  The only assignments to $q.items$
are $E_1$ and $D_2$: the former can only happen by an $\enq(x)$, which puts $x$ into
the array; the latter assigns $\NULL$.

Likewise, there are no \VRepet\ violations because whenever in an arbitrary execution trace two calls to
$\deq()$ return the same $x$, then at least twice there was an element of 
the $q.items$ array holding the value $x$ and was updated to $\NULL$ by
the $\texttt{SWAP}$ instruction at $D_2$.
Therefore, at least two assignments of the form $q.items[\_] \gets x$ happened; 
i.e.\ there were at least two $\enq(x)$ events in the induced history.

We move on to the more challenging third condition, \VOrd. 
We actually consider its equivalent reformulation, \POrd.
Fix a value $v_2$ and consider an execution trace $\tau$ where every method call enqueuing $v_2$ 
is preceded by some method call enqueuing some different value $v_1$ and there
are no $\deq()$ calls returning $v_1$ (there may be arbitrarily many concurrent
$\enq()$ and $\deq()$ calls enqueuing or dequeuing other values). 
The goal is to show that in this execution trace, no $\deq()$ return $v_2$.

Let us suppose there is a dequeue $d$ returning $v_2$, and try to derive a contradiction.
For $d$ to return $v_2$, it must have read $\mathit{range} \geq i_2$ such that
$q.items[i_2] = v_2$. So, $d$ must have read $q.back$ at $D_1$ after
$\enq(v_2)$ incremented it at $E_1$.

Since, $\enq(v_1) \prec_{h(\tau)} \enq(v_2)$, it follows that $\enq(v_2)$ will have read a larger
value of $q.back$ at $E_1$ than $\enq(v_1)$.  So, in particular, once $\enq(v_1)$
finishes, the following assertion will hold:
\begin{equation}
\tag{$*$}\label{eq:POrd-inv}
    \exists i_1 < q.back.~ q.items[i_1] = v_1 \land (\forall j < i_1.~ q.items[j] \neq v_2) 
\end{equation}
Note that since, by assumption, $v_1$ can never be dequeued, and any later
$\enq(v_2)$ can only affect the $q.items$ array at indexes larger than $i_1$,
\eqref{eq:POrd-inv} is an invariant.

Given this invariant, however, it is impossible for $d$ to return $v_2$, as in
its loop it will necessarily first have encountered $v_1$. Formally, to show this
we use the following loop invariant at the beginning of $\mathbf{for}$ loop
\[
      \exists i_1.~  i \leq i_1 < q.back \land q.items[i_1] = v_1 \land (\forall j < i_1.~ q.items[j] \neq v_2) 
\]
and \eqref{eq:POrd-inv} for the while loop. With these invariants, it is immediate
that the swap at line $D_2$ cannot read $v_2$.

% !TEX root = main.tex
\section{Checking the Conditions by Proving Program Divergence}
\label{sec:checking}

In this section, we reduce proving the absence of \VFresh, \VRepet\ and \VOrd\
violations to proving that certain programs always diverge.  Towards the end of
the section, we also discuss how the absence of \VWit\ violations might be automatically
checked for queue implementations whose {\deq} method may return \NULL.

Our proof technique relies heavily on instrumenting the $\deq()$ function
with a prophecy variable `guessing' the value that will be returned when
calling it.
That is, we construct a method, $\deq(v)$, such that the set of execution traces of
$\bigndchoice_{x \in \mathbb{N}\cup\{\NULL\}} \deq(x)$ 
is equal to the set of execution traces of $\deq()$, where $\ndchoice$ stands for
(demonic) non-deterministic choice: the set of traces of $\mathcal{T} \sqcup \mathcal{T}'$ is
the union of the sets of traces of $\mathcal{T}$ and $\mathcal{T}'$.
A simple construction is to define $\deq(v)$ to behave exactly as $\deq()$
except that when $\deq()$ is about to return a value other than $v$, 
we make $\deq(v)$ diverge.  That is, we prepend an $\mathsf{assume}(x = v);$
statement to every $\mathbf{return}\ x$ statement in $\deq()$.
In Section~\ref{sec:cave}, we describe a better construction.

\subsection*{Proving Absence of \VFresh\ Violations}

Generally, it is completely straightforward to prove the absence of \VFresh\
violations.  For example, it is sufficient for the queue implementation to be
data independent~\cite{WP1986}.

This is because a data independent implementation cannot produce values `out of
thin air.' In other words, if a dequeue returns a value, it must have read that
value from memory, and the only way for a value to get into memory is for an
enqueue to be invoked with that value passed as an argument.  Therefore, no
\VFresh\ violations can occur in data independent implementations.

\subsection*{Proving Absence of \VRepet\ Violations}

To prove the absence of \VRepet\ violations, we use the following theorem. 

\begin{thm}\label{thm:vrepet}
A completable queue implementation has no \VRepet\ violations iff 
for all values $v$ and all $n,m,k \in \mathbb{N}$ such that $0 < n < m$, the program
\[
 \Prg(v,n,m,k) \defeq (\overbrace{\enq(v) \| \ldots \| \enq(v)}^{n\text{ times}} \,\|\,
                     \overbrace{\deq(v) \| \ldots \| \deq(v)}^{m\text{ times}} \,\|\, 
                     \overbrace{C \,\|\, \ldots \,\|\, C}^{k\text{ times}} )
\]
has no execution trace in which more than $n$ $\deq(v)$ threads terminate, where
\[
  C \defeq   
   \bigndchoice_{x\neq v} \enq(x) \ndchoice 
   \bigndchoice_{x\neq v} \deq(x)
\,.
\]
\end{thm}

\begin{proof}
($\Rightarrow$)
We argue by contradiction.
Consider an execution trace $\tau$ of $\Prg(v,n,m,k)$ where at least $n+1$ of the $\deq(v)$ threads terminate.
The induced history $h(\tau)$ cannot have a safe matching because to satisfy condition (1) of
Definition~\ref{def:safe}, each $\deq(v)$ must be matched by some $\enq(v)$,
and from the pigeonhole principle multiple $\deq(v)$ will have to be matched
with the same $\enq(v)$, 
thereby violating condition (3) of the Definition.

($\Leftarrow$)
Again, we argue by contradiction.
Assume the queue implementation has an execution trace $\tau$ such that $h(\tau)$ has a \VRepet\ violation. 
For each value $v$, 
let $n_v$ be the number of invoked $\enq(v)$ operations in $\tau$ 
and $m_v$ be the number of invoked $\deq(v)$ operations.
Then, since there is a \VRepet\ violation, for some $v$ there are at least 
$n_v + 1$ completed $\deq(v)$ operations in $\tau$.
Finally, observe that $\tau$ can be generated by a run of the program $\Prg(v,n_v,m_v,k)$
(for some $k$) in which at least $n_v + 1$ of the $\deq(v)$ threads terminate.
\end{proof}

In case the queue implementation is data independent~\cite{WP1986}, we can
simplify the \VRepet\ check further.
We say that a history is \emph{differentiated}, if all the input arguments to invocations
of the library's methods are pairwise different.  Given a renaming function on data
values, $f : \mathcal{D} \to \mathcal{D}$, we write $f(c)$ for applying the function
to all the data values in the history $c$.
An implementation is \emph{data independent}, if the set of histories it generates, $H$,
satisfies two properties: (1) for every $c \in H$, $f(c) \in H$; and (2) for every $c \in H$,
there exists a differentiated history $c' \in H$ such that $c = f(c')$.
To ensure data independence, it suffices to check that the implementation
never performs any operations (such as testing for equality) on the value domain.

For data-independent programs, we can reduce reasoning about any number 
(say $n$ and $m$ where $m > n$) of $\enq(v)$ and
$\deq(v)$ threads to a single $\enq(v)$ and multiple $\deq(v)$ threads.
To see why a data independence condition is necessary,
consider the following incorrect $\enq(v)$ and $\deq()$ implementations:
\begin{align*}
\enq(v) \defeq \mathsf{atomic}\ (\mathsf{if}\ v \in Q\ \mathsf{then}\ Q := Q{\cdot}v{\cdot}v\ \mathsf{else}\ Q := Q{\cdot}v) \\
\deq() \defeq \mathsf{atomic}\ (\mathsf{match}\ Q\ \mathsf{with}\ \epsilon \to \mathsf{block} \mid v\cdot Q' \to Q:=Q'; \mathsf{return}\ v)
\end{align*}
Observe that for all $m > 1$, the program $\Prg(v,1,m,0)$ never terminates whereas
the program $\Prg(v,2,3,0)$ has a terminating execution: the serial execution where
both enqueues take place before all the dequeues.

\begin{thm}\label{thm:vrepet-di}
A data-independent completable queue implementation has no \VRepet\ violations
iff for all values $v$, all $m > 1$ and all $k \in \mathbb{N}$, the program
$\Prg(v,1,m)$ (as defined in Theorem~\ref{thm:vrepet})
has no execution in which more than one $\deq(v)$ threads terminate.
\end{thm}

\begin{proof}
By Theorem~\ref{thm:vrepet}, it suffices to show that if
for all $v$, $m$ and $k$, $\Prg(v,1,m,k)$ has no execution trace with more than one
terminating $\deq(v)$, then
for all $v$, $n$, $m$ and $k$, no execution trace of the program $\Prg(v,n,m,k)$ can have more
than $n$ terminating $\deq(v)$ threads.
Now, as $\enq$ and $\deq$ do not perform any value-dependent operations, we can
replace the $v$ being enqueued by distinct fresh $v_i$ values.
Doing so will naturally affect the return values of the dequeue operations that
were returning $v$, but because of data independence, nothing else.
Hence, the program 
\[
   \overbrace{\enq(v_1) \parallel \ldots \parallel \enq(v_n)}^{n\text{ threads}} \parallel 
   \overbrace{\deq(r_1) \parallel \ldots \parallel \deq(r_m)}^{m\text{ threads}} \parallel 
   \overbrace{C \parallel \ldots \parallel C}^{k\text{ times}}
\]
must have an execution trace where at least $n+1$ of the $\deq(r_i)$ threads
terminate with $r_i \in \{ v_1, \ldots, v_n \}$ for $0 < i < m$.
So, by the pigeonhole principle, there exists some value $v_i$ that gets 
dequeued multiple times, say $m'$.  This, however, contradicts our assumption
that $\Prg(v_i,1,m',-)$ has at most one terminating $\deq(v_i)$ thread.
\end{proof}

\subsection*{Proving Absence of \VOrd\ Violations}

We move on to the \POrd\ property, which as we have seen in the manual proof of
the HW queue, is often more complicated to prove.  It turns out that our
automated technique for proving \POrd\ also establishes absence of \VFresh\
violations as a side-effect.
We reduce the problem of proving absence of \VFresh\ and \VOrd\ violations to the
problem of checking non-termination of non-deterministic programs with an
unbounded number of threads.  The reduction exploits the instrumented
$\deq(v)$ definition: $\deq()$ cannot return a result $x$ in an execution
precisely if $\deq(x)$ cannot terminate in that same execution. 

\begin{thm}\label{thm:vord}
A completable queue implementation has no \VFresh\ and \VOrd\ violations iff 
for all $k \in \mathbb{N}$ and for all $v_1$ and $v_2$ such that $v_1\neq v_2$, the $\deq(v_2)$
thread does not terminate in the program
\[
   \Prg(k) \defeq b\gets\false; ~ (\deq(v_2) \,\|\, (\enq(v_1); b\gets\true) \,\|\,
                          \overbrace{C \,\| \ldots \|\, C}^{\text{$k$ threads}} )
\]
where
\[
  C \defeq
  (\assume(b); \enq(v_2)) \ndchoice 
   \bigndchoice_{x\neq v_2} \enq(x) \ndchoice 
   \bigndchoice_{x\neq v_1} \deq(x)
\,.
\]
\end{thm}

\begin{proof}
($\Rightarrow$)
We argue by contradiction.
Consider an execution trace $\tau$ of $\Prg(k)$ in which the $\deq(v_2)$ thread terminates.
If $\enq(v_2)$ is not invoked in $\tau$, then as there are no \VFresh\ violations, 
we know that no $\deq()$ in $\tau$ can return $v_2$, 
contradicting our assumption that $\deq(v_2)$ terminates in $\tau$.
Otherwise, if $\enq(v_2)$ is invoked in $\tau$, then at some earlier point
$\assume(b)$ was executed, and since initially $b$ was set to $\false$,
this means that $b\gets\true$ was executed and therefore $\enq(v_1) \prec_{h(\tau)} \enq(v_2)$.
Consequently, from \POrd, if there is $\deq()$ in $\tau$ returns $v_2$, there must
be a $\deq()$ in $\tau$ that can be completed to return $v_1$, contradicting
our assumption that $\deq(v_2)$ terminates in $\tau$.

($\Leftarrow$) We have two properties to prove.
For \VFresh, it suffices to consider the restricted parallel context that
never enqueues $v_2$. In this restricted context,
$\deq(v_2)$ does not terminate, and so $\deq()$ cannot return $v_2$.
For \VOrd, consider an execution trace in which every $\enq(v_2)$ happens after
some enqueue of a different value, say $\enq(v_1)$, and in which there is no 
$\deq(v_1)$.  Such an execution trace can easily be produced by the unbounded parallel
composition of $C$, and so $\deq(v_2)$ also does not terminate, as required.
\end{proof}

\subsection*{Showing Absence of \VWit\ Violations}

Here, we have to show that any dequeue event cannot return \texttt{NULL} if it never goes through a state where the queue could be logically empty.
This in turn means that we have to express non-emptiness using only the actions of the history (and not referring to the linearization point or the gluing invariant which relates the concrete states of the implementation to the abstract states of the queue).
For the following let us fix a (complete) concurrent history $c$ and a dequeue of interest $\dhat$ which returns {\NULL} and does not precede any other event in $c$.

Let $c'$ be some prefix of $c$ and let $e\in \Enq{c'}$ be a completed enqueue event in $c'$.
We will call $e$ {\em alive} after $c'$ if there is a matching dequeue event $d$ in $\Deq c$, i.e. $d=\deq(\Val {c} e)$, 
then $d$ is neither pending nor completed in $c'$. 
In other words, $e$ is alive after $c'$ if its matching dequeue $d$, if it exists, is not invoked in $c'$.

For the following, let $d_i$ denote the dequeue event which removes the element inserted by the enqueue event $e_i$; that is, $d_i=\deq(\Val{c}{e_i})$.
A sequence $e_0e_1\ldots e_n$ of enqueue events in $\Enq c$ is {\em covering} for $\dhat$ in $c$ if the following holds:
\begin{itemize}
\item $e_0$ is alive at $c'$ where $c'$ is the maximal prefix of $c$ in which $inv(\dhat)$ does not occur.
\item For all $i\in[1,n]$, $e_i$ starts before $\dhat$ completes.
\item For all $i\in[1,n]$, we have $e_i\prec_c d_{i-1}$.
\item $e_n$ is alive at $c$.
\end{itemize}
Note that all $d_i$ must exist by the third condition, with the only exception of $d_n$, which does not exist (the last condition).
Then, the sequence is covering for $\dhat$ if $d_0$ does not start before $\dhat$ starts, and every enqueue event $e_i$ completes before the dequeue event $d_{i-1}$ starts.
Intuitively, this means that at every state visited during the execution of $\dhat$, the queue contains at least one element. 

The property corresponding to the last violation (\VWit) then becomes the following:
\begin{description}
\item[(\PWit)] A dequeue event $d$ cannot return {\NULL} if there is a covering for $d$.
\end{description}

\begin{lem}\mylabel{lem:vwit-pwit}
A (complete) concurrent history $c$ has {\VWit} iff it does not satisfy {\PWit}.
\end{lem}
\begin{proof}
($\Rightarrow$)
Let $c$ have {\VWit}.
By Prop.~\ref{prop:compl-viol}, there is $\dhat\in \Deq c$ such that $\Val c \dhat=\NULL$ and $\Bad c {\dhat}\cap \Before c {\dhat}\neq \emptyset$.
We construct a covering sequence $e_0\ldots e_n$ for $\dhat$ such that for all $0\leq i< n$ the response of $e_i$ occurs before the response of $e_{i+1}$, if $j_i$ and $j_{i+1}$ are minimal indices for which $e_i\in \Badx c {\dhat} {j_i}$ and $e_{i+1}\in \Badx c {\dhat} {j_{i+1}}$ hold, then $j_{i+1}<j_i$, and $e_n\in \Badx c {\dhat} 0$, and if $e\in \Badx c {\dhat} k$ with $k<j_{i+1}$, then $e\not\prec_c d_i$.

\begin{description}
\item[(Base)] By the assumption there is an enqueue event in $\Bad c {\dhat}\cap \Before c {\dhat}$.
Set $e_0$ an enqueue event in $\Badx c {\dhat} {j_0}$ such that for any other enqueue event $e'\in \Badx c {\dhat} k\cap \Before c {\dhat}$, we have $j_0\leq k$.

\item[(Inductive)] Let $e_i$ be in $\Badx c {\dhat} {j_i}$ with $j_i>0$.
Let $E'$ be the set of all $e'\in \Bad c {\dhat}$ such that either $e'\prec_c e_i$ or $e'\prec_c d_i$, where $d_i$ is the matching dequeue event for $e_i$.
Observe that $E'$ is non-empty.
Choose $e_{i+1}\in E'$ to be an enqueue event with minimal index in $E'$.
That is, if $j_{i+1}$ is the smallest index for which $e_{i+1}\in \Badx c {\dhat} {j_{i+1}}$ holds, then for any $e'\in E'$, $e' \in \Badx c {\dhat} k$ implies $j_{i+1}\leq k$.
Observe that $j_{i+1}<j_i$.
This implies that by construction it cannot be the case that $e_{i+1}\prec_c d_{i-1}$ since it would contradict the assumption that $e_i$ was chosen as an enqueue event with minimal index among those that precede $d_{i-1}$.
But again by construction we have $e_i\prec_c d_{i-1}$ which implies that the response event of $e_{i+1}$ occurs after the response event of $e_i$.
This also means that because $e_{i+1}\not\prec_c e_i$, we must have $e_{i+1}\prec_c d_i$. 

\end{description}
Since the sequence of indices $j_i$ is strictly decreasing, to show that the construction terminates with $j_n=0$, we only have to show that there is $e_n\in \Badx c {\dhat} 0$ completed before $\dhat$ is completed; i.e. the response of $e_n$ occurs before the response of $\dhat$ in $c$.
By the definition of {\VWit}, taking $c_p=c_0\cdot inv(\dhat)\cdot c_d$, we know that there must be at least one enqueue event $e$ in $c_p$ such that $e$ is completed in $c_p$ and its matching dequeue is neither pending nor completed in $c_p$. 
But this immediately implies that $e\in \Badx c {\dhat} 0$ and $e$ is completed before $\dhat$ is completed.

($\Leftarrow$)
Let $e_0\ldots e_n$ be a covering sequence for $\dhat$.
Then, $e_n\in \Bad c {\dhat}$ because $d_n$ if it exists is preceded by $\dhat$, i.e. $\dhat\prec_c d_n$.
Furthermore, for every $i\in[1,n]$, since we have $e_i\prec_c d_{i-1}$, all $e_i\in \Bad c {\dhat}$.
Finally, $e_0\in \Before c {\dhat}$. 
Thus, $\Bad c {\dhat}$ and $\Before c {\dhat}$ are not disjoint if there is a covering for $\dhat$.
By Prop.~\ref{prop:compl-viol} this implies the existence of {\VWit}.
\end{proof}

\noindent
We will actually restate the same property in a simpler way by making the
following observation.
\begin{prop}
There is a covering for $\dhat$ in $c$ iff at every prefix $c'$ of $c$ such that
$\dhat$ is pending in $c'$, there is at least one alive enqueue event.
\end{prop}

\noindent
Then, we can alternatively state {\PWit} as follows:
\begin{description}
\item[($\PWit'$)] A dequeue event $d$ cannot return {\NULL} if for every prefix $c'$ at which $d$ is pending there exists an alive enqueue event.
\end{description}
Note that {\POrd} can also be stated in terms of alive enqueue events. 
\begin{description}
\item[($\POrd'$)] For any enqueue events $e_1$ and $e_2$ with $e_1 \prec_c e_2$ and \mbox{$\Val{c}{e_1}\neq\Val{c}{e_2}$},
a dequeue event cannot return $\Val{c}{e_2}$ if $e_1$ is alive at $c$.
\end{description}

% !TEX root = main.tex
\section{Automation within Cave}
\label{sec:cave}

\begin{figure}[t]
\begin{algorithmic}
 \Procedure{$\deq$}{$v : val$}
  \While{true}
   \State $\inatom{ range\gets q.back-1 }$
  \For{$i=0$ \textbf{to} $range$}
  \State 
$\inpar{\inatom{
    x\gets q.items[i] ;
  \\ \assume(x = v \land x \neq \NULL) ;
  \\ q.items[i] \gets \NULL } ;
  \\ \mathbf{return}\ x}
\ndchoice
\inatom{
   x\gets q.items[i] ;
\\ \assume(x = \NULL) ;
\\ q.items[i] \gets \NULL
}
$
  \EndFor
 \EndWhile
 \EndProcedure
\end{algorithmic}
\caption{The HW dequeue method instrumented with the prophecy variable $v$ guessing
its return value, where $\ndchoice$ stands for non-deterministic choice.}
\label{fig:instr-deq}
\end{figure}

To automate the linearizability proof of the HW queue, we have mildly adapted
the implementation of \textsc{Cave}~\cite{Vaf2010}, a sound but incomplete
thread-modular concurrent program verifier that can handle dynamically
allocated linked list data structures and fine-grained concurrency.
The tool takes as its input a program consisting of some initialization code
and a number of concurrent methods, which are all executed in parallel an unbounded
number of times each. When successful, it produces a proof in RGSep that the
program has no memory errors and none of its assertions are violated at runtime.
Internally, it performs RGSep action inference~\cite{Vaf2010a} with a rich
shape-value abstract domain~\cite{Vaf2009} that can remember invariants 
indicating that value $v_1$ is inside a linked list.
\textsc{Cave} also has a way of proving linearizability by a brute-force search
for linearization points (see~\cite{Vaf2010} for details), but this is not
applicable to the HW queue and therefore irrelevant for our purposes.

\subsection*{Overview of Action Inference}

In brief, \textsc{Cave}'s action inference algorithm first determines the part of
the heap-allocated memory that is private to a thread and the part that is shared. 
The main heuristic employed in this decision is that newly allocated memory 
cells are deemed to be private until they become reachable from some global 
variable, from which point onwards they are deemed shared.

Next, the algorithm computes a binary relation $R$ on program states 
overapproximating the effects of all atomic statements of the program 
to the shared part of the heap. 
Syntactically, it represents $R$ as the union of a set of more primitive 
binary relations, which are called \emph{actions}.
Moreover, it remembers which atomic program statements correspond 
to which actions of the set.
Thus, for example, if we want to compute an overapproximation of a 
program $C$ in a parallel context, $C'$, we can run action inference
on $C\|C'$ and from the total set of actions return only those corresponding
to $C$. 

As part of this overapproximation, any information about the program's 
control flow is lost except when the program explicitly records it in some 
global variable.  This property is common to most thread-modular 
reasoning techniques, and is necessary for scalability.
Thus, for instance, the programs $C$, $C^*$, and $C\|C$ generate 
the same set of actions.

In the process of computing the set of actions, \textsc{Cave} 
proves that the program is memory safe and does not violate any 
assertions in it. To do so, it constructs a proof in RGSep, which is an
adaptation of Jones' rely-guarantee method suitable for pointer-manipulating 
programs~\cite{Jon1983,VP2007}. 
To construct these proofs, it calculates via abstract interpretation
an invariant that holds after every atomic program statement.
These invariants describe the shapes of the heap allocated data structures
(e.g., that there is a linked list from $x$ to $y$ via the field \texttt{next}),
and some very simple facts about the values stored in them 
(e.g., that the sequences of values stored in two list segments are equal,
or that the sequence of values stored in one list segment is sorted).

Finally, we note that action inference is incremental. Typically, action inference 
is run starting with an initial empty set of actions, to which set it adds any new
actions it generates until a fixpoint is reached. 
When, however, we want to verify $C\|C'$ and we already know a sound 
abstraction of $C$ (under the assumption that $C'$ can be run in parallel), 
it suffices to perform action inference only on $C'$ but starting with the 
set of actions of $C'$ as the initial set of actions.
To this set, action inference will add any further actions $C$ produces.

\subsection*{Summary of Changes}

The modifications we had to perform to \textsc{Cave} were:
\begin{enumerate}
\item To add code that instruments $\deq()$ methods with a prophecy argument
guessing its return value, thereby generating $\deq(v)$;
\item To add some glue code that constructs the verification conditions of
Theorems~\ref{thm:vrepet-di} and~\ref{thm:vord} and runs the 
underlying prover to verify them;
\item To improve the abstraction function so that it can remember properties of the
form $v_2 \notin X$, which are needed to express the \eqref{eq:POrd-inv} invariant
of the proof in Section~\ref{sec:herlihy-wing}; and
\item When checking the absence of $\VRepet$ violations, 
to instrument the inferred actions so as to work around the fact 
that action inference abstracts over control flow information.
\end{enumerate}
The first two changes are clearly tool-independent, the third item is 
very \textsc{Cave}-specific, whereas the fourth item is fairly generic. 
The problem that we are working around here is common to almost
all thread-modular verification approaches, and our instrumentation 
should work for other tools as well.  
To use a different tool from \textsc{Cave}, the tool must be able to express 
invariants such as the aforementioned \eqref{eq:POrd-inv} invariant.

As \textsc{Cave} does not support arrays (it only supports linked lists), we
gave the tool a linked-list version of the HW queue, for which it successfully 
verified that there are no \VFresh, \VRepet, and \VOrd\ violations.
(As the HW deques never return $\NULL$, the algorithm also trivially has no
\VWit\ violations.)

\subsection*{Prophetic Instrumentation of Dequeues}

In order to be able to use the theorems in the previous section, we must first
construct the method $\deq(v)$ that records the result of the $\deq()$ function
in its arguments which acts like a prophecy variable.
In essence, the $\deq(v)$ we construct must be such that the set of traces of
$\bigndchoice_{x \in \mathbb{N}\cup\{\NULL\}} \deq(x)$ 
is equal to the set of traces of $\deq()$, where $\ndchoice$ stands for
non-deterministic choice.
Figure~\ref{fig:instr-deq} shows the resulting automatically-generated 
instrumented definition of $\deq(v)$ for the HW queue.

Our implementation of the instrumentation performs a sequence of simple rewrites,
each of which does not affect the set of traces produced:
\begin{align*}
\mathbf{return }\ E &\leadsto  \assume(v = E); \mathbf{return}\ E \\
\mathbf{if}\ B\ \mathbf{then}\ C\ \mathbf{else}\ C' 
                &\leadsto  (\assume(B); C) \ndchoice (\assume(\lnot B); C')  \\
C ; \assume(B)  &\leadsto  \assume(B); C    
                \qquad\qquad\text{provided } \mathit{fv}(B) \subseteq \mathit{Locals} \setminus \mathit{writes}(C) \\
C ; (C_1 \ndchoice C_2)  &\leftrightsquigarrow   (C ; C_1) \ndchoice (C; C_2) \\
(C_1 \ndchoice C_2) ; C  &\leftrightsquigarrow   (C_1; C) \ndchoice (C_2; C) 
\end{align*}
In general, the goal of applying these rewrite rules is to bring the introduced
$\assume(v = E)$ statements as early as possible without unduly duplicating
code.

\subsection*{Instrumentation for Checking Absence of \VRepet\ Violations}

Observe that the HW queue implementation is data independent as the operations 
on the shared locations in the $\enq$ and $\deq$ methods do not depend
on the value of argument. 
Therefore, using Theorem~\ref{thm:vrepet-di}, we have to prove that in the
context where only one $\enq(v)$ can happen in parallel, $\deq(v)$ cannot
terminate if another $\deq(v)$ has terminated.

One slight complication is that we cannot use RGSep action
inference~\cite{Vaf2010} directly to prove this property because we have to
keep track of the exact number of occurences of particular shared memory
operation (such as the enqueues of $v$).
In rely-guarantee, operations on shared variables are abstracted by 
\emph{actions}, which typically do not contain any control flow within them. 
Hence after the initial action generation, 
we have to augment the shared state and the actions with 
auxiliary variables that 
(a) record the termination of parallel $\deq(v)$ and 
(b) ensure that only one parallel $\enq(v)$ call is accounted for.
Our implementation therefore proceeds as follows:
\begin{enumerate}
\item It infers an initial set of RGSep actions, $R$, by performing symbolic
execution of the {\enq} and {\deq} methods, and refine this set of actions
to record information about the arguments of $\enq()$ and the result of the
$\deq()$ functions wherever possible.
Let $R_\enq$ be the actions generated by $\enq$ method and $R_\deq$ be those
generated by $\deq$.

\item For each action that is executed at most once by an $\enq(v)$ invocation,
it generates a fresh auxiliary variable, $e_i$, and records that $e_i$ changes
from $0$ to $1$ by performing that action.
Formally, we define:
\[\begin{array}{r@{~}l}
E  \defeq & \{ (\ell, A) \in R_{\enq} \mid \ell \text{ occurs at most once on every path through $\enq$} \} \\
R' \defeq & \{ (\ell, A \land e_{\ell} = 0 \land e'_{\ell} = 1) \mid (\ell, A) \in E \}
\cup (R_{\enq} \setminus E) \,.
\end{array}\]
writing $e_{\ell}$ and $e'_{\ell}$ for the freshly generated variables in the
action's pre- and post-states.
(The purpose of this instrumentation is to ensure that the $E$ actions will 
not interfere more than once with $\deq(v)$ below.)

\item Record each action that must be performed by a completed $\deq(v)$ event
using a fresh auxiliary variable, $d_i$.  Formally,
\[\begin{array}{r@{~}l}
D   \defeq & \{(\ell,A) \in R_\deq  \mid \ell \text{ must occur on every path through $\deq$} \} \\
R'' \defeq & \{(\ell,A \wedge d'_{\ell} = 1) \mid (\ell, A)\in D\} \cup (R_{\deq} \setminus D) \,. 
\end{array}\]
where $d'_{\ell}$ are the freshly generated variables in the action's
post-state.
(The purpose of this instrumentation is to be able to detect whether a
$\deq$ operation has terminated.) 

\item 
Running action inference with the following initial set of actions (the rely condition)
\[
  R'[v/\mathit{arg}] \cup R''[v/\mathit{res}] \cup 
\bigcup_{v'\neq v} (R_{\enq}[v'/\mathit{arg}] \cup R_{\deq}[v'/\mathit{res}]) 
\,,
\]
verify the Hoare triple
\[
   \left\{ e_1 \,{=} \ldots{=}\, e_n \,{=}\, d_1 \,{=}\ldots {=}\, d_m \,{=}\, 0 \right\}\ \deq(v)\
   \left\{ \exists i.~ d_i = 0 \right\}
\,.
\]
The postcondition ensures that no other $\deq(v)$ has terminated, because if it
had, it must have set each $d_i = 1$.
\end{enumerate}

%%% Local Variables: 
%%% mode: latex
%%% TeX-master: "main"
%%% End: 

%!TEX root = main.tex

\section{Related Work}
\label{sec:related-work}

Linearizability was first introduced by Herlihy and Wing~\cite{HW1990}, who
also presented the HW queue as an example whose linearizability cannot be proved
by a simple forward simulation where each method performs its effects
instantaneously at some point during its execution.
The problem is, as we have seen, that neither of $E_1$ or $E_2$ can be given
as the (unique) linearization point of $\enq$ events, because the way in which
two concurrent enqueues are ordered may depend on not-yet-completed concurrent
$\deq$ events.  In other words, one cannot simply define a mapping from the
concrete HW queue states to the queue specification states. 
Nevertheless, Herlihy and Wing do not dismiss the linearization point technique
completely, as we do, but instead construct a proof where they map concrete
states to non-empty sets of specification states.  

This mapping of concrete states to non-empty sets of abstract states is closely
related to the method of \emph{backward simulations}, employed by a number of
manual proof efforts~\cite{CDG2005,DM2009,SWD2012}, 
and which Schellhorn et al.~\cite{SWD2012} recently showed to be a complete
proof method for verifying linearizability.
Similar to forward simulation proofs, backward simulation proofs, are monolithic 
in the sense that they prove linearizability directly by one big proof.
Sadly, they are also not very intuitive and as a result often difficult to come
up with. For instance, although the definition of their backward simulation
relation for the HW queue is four lines long, Schellhorn et al.~\cite{SWD2012}
devote two full pages to explain it.

As a result, most work on automatically verifying linearizability
(e.g.~\cite{ARR+2007,Vaf2009,Vaf2010,AHH+2013,DGH2013}) and some
manual verification efforts (e.g.,~\cite{DSW2011,CDG2005}) have relied on the simpler
technique of forward simulations, even though it is known to be incomplete.
The programmer is typically required to annotate each method with its
linearization points and then the verifier uses some kind of shape analysis
that automatically constructs the simulation relation. 
This approach seems to work well for simple concurrent algorithms such as the
Treiber stack and the Michael and Scott queues, where finding the linearization
points may be automated by brute-force search~\cite{Vaf2010}.
Most recently, with their technique based on (automatically) rewriting implementations Dragoi et al.~\cite{DGH2013} have succeeded to extend this approach to some implementations with helping.
Similar to their precursors, however, their approach also assumes the existence of static linearization points, i.e.\ instructions in the program code that when executed invariably correspond to the linearization of one or more methods.
Thus, there are many implementations, as mentioned in the Introduction, that cannot be handled by this approach.

Among this line of work, the most closely related one to this paper is the
recent work by Abdulla et al.~\cite{AHH+2013}, 
who verify linearizability of stack and queue algorithms using observer
automata that report specification violations such as our \VOrd. 
Their approach, however, still requires users to annotate methods
with linearization points, because checker automata are synchronized with the
linearization points of the implementation.

To the best of our knowledge, there exist only two earlier published proofs 
of the HW queue: (1) the original pencil-and-paper proof by Herlihy and
Wing~\cite{HW1990}, and (2) a mechanized backward simulation proof by
Schellhorn et al.~\cite{SWD2012}.

Both proofs are manually constructed. In comparison,  our new proof 
is simpler, more modular, and automatically generated.
This is largely due to the fact that we have decomposed the goal of proving
linearizability into proving four simpler properties,
which can be proved independently.
This may allow one to adapt the HW queue algorithm, e.g.\ by checking
emptiness of the queue and allowing $\deq$ to return \NULL, and affecting
only the proof of absence of \VWit\ violations without affecting the
correctness arguments of the other properties. 

Our violation conditions are arguably closer to what programmers have
in mind when discussing concurrent data structures.  Informal specifications 
written by programmers and bug reports do not mention that some method is 
not linearizable, but rather things like that values were dequeued in the wrong
order.

\section{Conclusion}
\label{sec:conclusion}

We have presented a new method for checking linearizability of concurrent queues.
Instead of searching for the linearization points and doing a monolithic simulation proof, we verify four simple properties whose conjunction is equivalent to linearizability with respect to the atomic queue specification.
By decomposing linearizability proofs in this way, we obtained a simpler correctness proof of the Herlihy and Wing queue~\cite{HW1990}, and one which can be produced automatically.

We believe that our new property-oriented approach to linearizability proofs
will be applicable to other kinds of concurrent shared data structures, such as
stacks, sets, and maps.
The generalization, however, is not entirely straightforward.  In the case of
stacks, the violations are similar to that of queues, but not exactly dual. 
The main difference is that the ordering violation for stacks is similar to
\VWit\ and not to \VOrd\ as one might expect.
Similarly, the violations for set implementations are also not as simple as
dropping the ordering constraint.  Instead, we need to count the number of
successful insertions and deletions to express what can go wrong.  It remains
to be seen, however, whether such counting arguments can yield an automatic
verification technique. 

%%% Local Variables: 
%%% mode: latex
%%% TeX-master: "main"
%%% End: 

\section*{Acknowledgments}
We would like to thank the CONCUR'13 reviewers for their feedback. 
The research was supported by the EC FET FP7 project ADVENT, by the Austrian
Science Fund NFN RISE (Rigorous Systems Engineering), by the ERC Advanced
Grant QUAREM (Quantitative Reactive Modeling), and by the EPSRC Grants EP/H005633/1 and EP/K008528/1.

\bibliographystyle{abbrv}
\bibliography{biblio}

\end{document}